\documentclass[acmsmall,screen,nonacm]{acmart}
\AtBeginDocument{%
  \providecommand\BibTeX{{%
    \normalfont B\kern-0.5em{\scshape i\kern-0.25em b}\kern-0.8em\TeX}}}

\setcopyright{none}
\copyrightyear{2024}
\acmYear{2024}
\acmDOI{XXXXXXX.XXXXXXX}

\DeclareMathOperator*{\argmin}{arg\,min}

\acmJournal{JACM}
\acmVolume{37}
\acmNumber{4}
\acmMonth{8}

\usepackage[frozencache,cachedir=.]{minted}
\usepackage{sourcecodepro}
\setminted{fontsize=\footnotesize}
\usepackage{mdframed}
\usepackage{tikz}

\usepackage{microtype}
\usepackage{graphicx}
\usepackage{wrapfig}
\usepackage{amsfonts}
\usepackage{listings}
\usepackage{amsmath}
\usepackage{subfigure}
\usepackage{booktabs} % for professional tables
\usepackage{multirow}
\usepackage{enumitem}
\usepackage{caption}

\usepackage{amsthm}

\newtheorem{definition}{Definition}
\newtheorem{theorem}{Theorem}

\newcommand{\ds}{ACSR}

\newcommand{\cublas}{cuBLAS}
\newcommand{\cusparse}{cuSPARSE}
\newcommand{\sddmm}{R-SDDMM}
\newcommand{\spmm}{R-SpMM}
\newcommand{\triton}{Triton}

\usepackage[ruled,linesnumbered]{algorithm2e}
\usepackage{graphicx}

\graphicspath{ {./images/} }

\begin{document}

%%
%% The "title" command has an optional parameter,
%% allowing the author to define a "short title" to be used in page headers.
\title{SPLAT: A framework for optimised GPU code-generation for SParse reguLar ATtention}

%%
%% The "author" command and its associated commands are used to define
%% the authors and their affiliations.
%% Of note is the shared affiliation of the first two authors, and the
%% "authornote" and "authornotemark" commands
%% used to denote shared contribution to the research.
\author{Ahan Gupta}
\affiliation{%
  \institution{UIUC}
  \city{Urbana}
  \state{Illinois}
  \country{USA}
}
\author{Yueming Yuan}
\affiliation{%
  \institution{UIUC}
  \city{Urbana}
  \state{Illinois}
  \country{USA}
}
\author{Devansh Jain}
\affiliation{%
  \institution{UIUC}
  \city{Urbana}
  \state{Illinois}
  \country{USA}
}
\author{Yuhao Ge}
\affiliation{%
  \institution{UIUC}
  \city{Urbana}
  \state{Illinois}
  \country{USA}
}
\author{David Aponte}
\affiliation{%
  \institution{Microsoft}
  \city{Seattle}
  \state{Washington}
  \country{USA}
}
\author{Yanqi Zhou}
\affiliation{%
  \institution{Google DeepMind}
  \city{San Francisco}
  \state{California}
  \country{USA}
}
\author{Charith Mendis}
\affiliation{%
  \institution{UIUC}
  \city{Urbana}
  \state{Illinois}
  \country{USA}
}
%\authornote{Both authors contributed equally to this research.}
%\email{trovato@corporation.com}
%\orcid{1234-5678-9012}
%\author{G.K.M. Tobin}
%\authornotemark[1]
%\email{webmaster@marysville-ohio.com}
%\affiliation{%
%  \institution{Institute for Clarity in Documentation}
%  \streetaddress{P.O. Box 1212}
%  \city{Dublin}
%  \state{Ohio}
%  \country{USA}
%  \postcode{43017-6221}
%}

%%
%% By default, the full list of authors will be used in the page
%% headers. Often, this list is too long, and will overlap
%% other information printed in the page headers. This command allows
%% the author to define a more concise list
%% of authors' names for this purpose.
% \renewcommand{\shortauthors}{Trovato and Tobin, et al.}

%%
%% The abstract is a short summary of the work to be presented in the
%% article.
\begin{abstract}

Multi-head-self-attention (MHSA) mechanisms achieve state-of-the-art (SOTA) performance across natural language processing and vision tasks. However, their quadratic dependence on sequence lengths has bottlenecked inference speeds. To circumvent this bottleneck, researchers have proposed various sparse-MHSA models, where a subset of full attention is computed. Despite their promise, current sparse libraries and compilers do not support high-performance implementations for \emph{diverse} sparse-MHSA patterns due to the underlying sparse formats they operate on. These formats, which are typically designed for high-performance \& scientific computing applications, are either curated for extreme amounts of random sparsity (<1\% non-zero values), or specific sparsity patterns. However, the sparsity patterns in sparse-MHSA are moderately sparse (10-50\% non-zero values) and varied, resulting in existing sparse-formats trading off generality for performance. 

We bridge this gap, achieving both generality and performance, by proposing a novel sparse format: affine-compressed-sparse-row (ACSR) and supporting code-generation scheme, SPLAT, that generates high-performance implementations for diverse sparse-MHSA patterns on GPUs. Core to our proposed format and code generation algorithm is the observation that common sparse-MHSA patterns have uniquely regular geometric properties. These properties, which can be analyzed just-in-time, expose novel optimizations and tiling strategies that SPLAT exploits to generate high-performance implementations for diverse patterns. To demonstrate SPLAT's efficacy, we use it to generate code for various sparse-MHSA models, achieving geomean speedups of 2.05x and 4.05x over hand-written kernels written in \triton{} and TVM respectively on A100 GPUs. Moreover, its interfaces are intuitive and easy to use with existing implementations of MHSA in JAX.
\end{abstract}

%%
%% The code below is generated by the tool at http://dl.acm.org/ccs.cfm.
%% Please copy and paste the code instead of the example below.
%%




\received{20 February 2007}
\received[revised]{12 March 2009}
\received[accepted]{5 June 2009}

%%
%% This command processes the author and affiliation and title
%% information and builds the first part of the formatted document.
\maketitle

\section{Introduction}
\label{Introduction}
Transformers have been widely adopted across various domains, powering popular applications like ChatGPT, Gemini Pro, and Claude, which handle millions of queries per day \cite{scale}. To effectively train and serve models at this scale, transformers must: (1) have high model quality by accurately answering user queries, and (2) utilize manycore GPU architectures effectively by repeatedly executing high-throughput kernels on these machines. However, achieving both simultaneously is challenging as datasets and tasks \cite{lra, scrolls} are demanding increasingly longer input sequences. This increases the memory consumption of multi-head-self-attention (MHSA) layers, which increases quadratically with respect to input sequence lengths and reduces the largest permissible batch size (LPBS) a model can operate on. Since sequences across batches are independent and can be processed in parallel, large models operating on long contexts are forced to use small batches and do not realize their high-throughput potential despite being embarrassingly parallel. 

To mitigate the memory bottleneck of MHSA, researchers have proposed several sparse-MHSA methods \cite{longformer, sparse-transformer, reformer, gemma-two, mistral}. These methods compute a subset of the entire attention matrix using a statically fixed mask. However, unlike the sparsity levels encountered in widely studied scientific and high-performance computing applications~\cite{SuiteSparse, sparsebench} which are extremely sparse (<1\% of the values are non-zero), state-of-the-art (SOTA) sparse-MHSA methods are \textit{moderately} sparse, computing 10-50\% of the full attention matrix. Computing fewer values degrades model quality while computing more values consumes additional memory, reducing the LPBS of models \cite{gemma-two, mistral}. The moderate sparsity ranges of sparse-MHSA place unique challenges in adopting existing sparse formats 
%and GPU code generation algorithms used in the scientific \& high-performance computing domains 
to implement high-performance kernels. 

\begin{wrapfigure}[19]{R}{0.5\textwidth}
    \includegraphics[width=\linewidth]{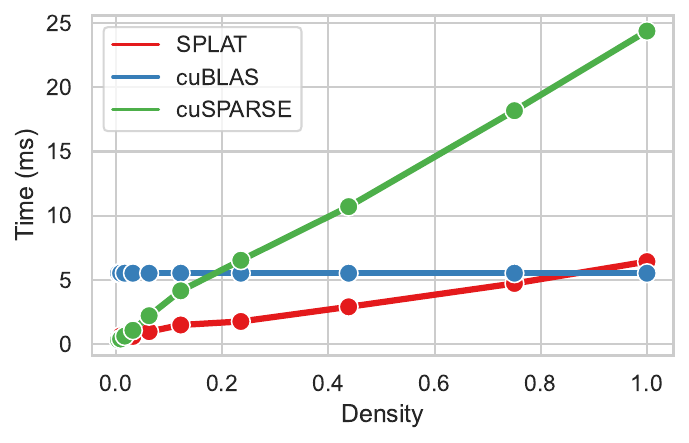}
    \caption{Run-time results for a sparse primitive used in sparse-MHSA (\spmm{})  comparing \cusparse{}, \cublas{} and SPLAT. We vary the density of the sparse input across: [0.4, 0.8, 1.6, 3, 6, 12, 24, 44, 75, 100]. The sparse input takes the shape of the blocked pattern (figure \ref{fig:regular-sparsity} right).}
    \label{fig:teaser-figure}
\end{wrapfigure}

On one end, general sparse libraries \cite{taco, loop-carried-sparse-dependencies, register-tiling-sparsity, cusparse} employ \textit{general sparse formats} (GSFs) that are designed to store extremely sparse matrices (<10\% of the values are non-zero). As a result, such formats, like the compressed-sparse-row (CSR) and coordinate (COO) formats, are augmented with metadata that represents the dense coordinates of each non-zero value and consumes memory in $O(f(nnzs))$. To correctly index a non-zero value stored in a GSF, its respective metadata must additionally be brought up the memory hierarchy, doubling the data read from high-bandwidth memory within the inner loops of sparse primitives. However, since sparse-MHSA layers are moderately sparse with SOTA approaches producing megabytes of non-zero values per layer \cite{mistral, gemma-two}, doubling the data read from high bandwidth memory (HBM) exacerbates contention of scarce L1 caches and register-file resources of the GPU, inhibiting per-iteration performance. As we see in figure \ref{fig:teaser-figure}, even hand-optimized vendor libraries like \cusparse{} \cite{cusparse} that employ the CSR format are outperformed by their dense counterparts, \cublas{} \cite{cublas}, for density levels as low as 20\% despite doing 1/5th of the compute.

On another end, to gain sizable latency and throughput benefits over full-attention, practitioners opt to hand-write specialized kernels with \textit{custom data formats} (CSFs). Such formats are designed with particular sparsity structures in mind, permitting format designs with lower metadata storage. Moreover, practitioners can analyze the sparsity structure to implement favorable thread access patterns that enhance cache reuse, memory coalescing, and launch kernels with high occupancy to facilitate inter-warp latency hiding. For example, triton's block-sparse kernels \cite{block-sparse} use blocks of predefined sizes as well as a look-up-table to map a block to the set of dense coordinates it operates on. Nevertheless, such an approach couples the performance of kernels to particular patterns at particular sparsity ranges. Attempting to adapt a kernel to a different pattern results in redundant storage and reduces the LPBS of models. 
%\charith{do we have numbers for another pattern showing it is suboptimal?} 
Unlike the GSFs employed by general sparse libraries, CSFs lack generality and do not scale in the fast-evolving field of sparse-MHSA where new patterns \cite{sparse-transformer, long-lora, mistral, gemma-two} rapidly emerge.

%% This is the older paragraph.
%On another end, to gain sizable latency and throughput benefits over full-attention, practitioners opt to hand-write specialized kernels for each sparsity pattern rather than use general sparse libraries. This specialization permits kernel writers to create custom sparse formats (CSFs) suitable to the unique indexing of particular sparsity patterns, reducing the metadata storage. Moreover, practitioners can analyze the sparsity structure to implement favorable thread access patterns that enhance cache reuse, memory coalescing, and launch kernels with high occupancy to facilitate inter-warp latency hiding. For example, triton's block-sparse kernels \cite{block-sparse} use blocks of predefined sizes as well as a look-up-table to map a block to the set of dense coordinates it operates on, similar to the block-compressed-sparse-row (BCSR) \textbf{cite} format. Nevertheless, such an approach couples the performance of kernels to particular patterns at particular sparsity ranges. Attempting to adapt a kernel to a different pattern results in redundant storage, and reduces the LPBS of models. Unlike general sparse libraries, hand-written kernels lack generality and do not scale in the fast-evolving field of sparse-MHSA where new patterns \cite{sparse-transformer, long-lora, mistral, gemma-two} rapidly emerge.

We observe that no general data format facilitates high-performance implementations for various sparse-MHSA patterns. GSFs require $O(f(nnzs))$ metadata storage to permit generality at the cost of performance while CSFs reduce metadata storage to permit performance at the cost of generality. We plug this gap and propose a novel data-format: affine-compressed-sparse-row (ACSR) and a GPU code generation framework, SPLAT (\textbf{SP}arse regu\textbf{L}ar \textbf{AT}tention), that produces high-performance sparse-MHSA implementations for a variety of sparse-MHSA patterns. Core to our methodology is the observation that sparse-MHSA patterns are regular. We mathematically formalize \textit{regular sparsity} and observe a novel geometric property of regularly sparse inputs: affine-compressibility. We exploit this property to design the \ds{} format with compressed metadata storage in $O(rows)$, an asymptotic reduction over $O(f(nnzs))$ compared to GSFs, while precisely representing \emph{diverse} sparse-MHSA patterns compared to CSFs. Consequently, the unique design of the \ds{} and layout in memory requires novel arithmetic and thread-access patterns to reference data, rendering many conventional sparse-optimizations that operate on existing formats \cite{loop-carried-sparse-dependencies, aspt, sparse-gpu-kernels-dl} in-applicable.
%Consequently, the unique design of the \ds{} and layout in memory differs from existing sparse-formats and requires novel arithmetic and thread-access patterns to reference data. As a result, rendering many conventional sparse-optimizations that operate on existing formats \cite{loop-carried-sparse-dependencies, aspt, sparse-gpu-kernels-dl} in-applicable. 
Therefore, we introduce novel GPU optimizations in SPLAT that collectively reason about thread divergence, memory coalescing, redundant compute, and cache reuse to generate high-performance implementations of sparse primitives that operate on the \ds{}.
%present in sparse-MHSA %Moreover, affine-compressibility exposes novel optimizations and tiling opportunities that SPLAT exploits with strong guarantees on thread divergence, memory coalescing, redundant compute, and cache reuse. 

To demonstrate SPLAT's efficacy, we implement 3 widely used sparse-MHSA patterns at various sparsity levels. SPLAT-generated SDDMM and SpMM kernels, two core primitives of sparse-MHSA, outperform highly-optimized vendor libraries cuSPARSE and cuBLAS at moderate sparsity levels by 2.81 \& 5.61x respectively. Moreover, SPLAT's end-to-end generated sparse-MHSA outperforms handwritten kernels in triton and TVM by up to 2.05x and 4.05x respectively. 

In summary, this paper makes the following contributions: 
\begin{itemize}[noitemsep, nolistsep]
    \item We investigate the geometric properties of sparse-MHSA patterns and observe their \textit{regularity}. We leverage this to propose a novel sparse-format: affine-compressed sparse-row (ACSR) to represent regularly sparse data with asymptotically low metadata storage requirements (Sections \ref{Overview} \& \ref{acsr}). 
    \item We develop novel optimized GPU code-generation schemes for regularly sparse kernels, what we term as \sddmm{} and \spmm{} kernels, that use the ACSR format. These code-generation schemes introduce novel tiling mechanisms with \textit{strong theoretical guarantees} on redundant compute, thread-divergence, memory access coalescing, and cache reuse for \sddmm{} kernels as well as novel optimizations that reduce redundant compute and produce favorable memory access/write patterns for \spmm{} kernels (Sections \ref{section:sddmm} \& \ref{spmm}).
    \item We use the optimized sparse operations to provide GPU code-generation strategies for end-to-end sparse-MHSA that reason about data-layout conversion costs between kernels to achieve globally efficient models (Section \ref{splat}). 
    \item We implement these code generation schemes in a framework which we name SPLAT. We conduct an extensive evaluation comparing a variety of sparse-MHSA models implemented in SPLAT to hand-optimized kernels released in Triton and TVM, achieving geomean speedups of 2.05x and 4.05x respectively, across a variety of sparsity levels. This shows that SPLAT can simultaneously achieve both generality and high-performance (Section \ref{Evaluation}).
\end{itemize}

\section{Background}
\label{Background}
\textbf{Full Attention} The backbone of the transformer is multi-head-self-attention (MHSA) \cite{full-attention}. MHSA computes the following matrix: $Concat(Head_1, Head_2, ..., Head_h)$ where $Head_i$ is: \\ \noindent$\underbrace{\emph{softmax}(QW_i^Q(KW_i^K)^T)}_{A_i}VW_i^V \in \mathbb{R}^{N \times d_m/h}$ and the concatenation happens across the columns of $A_iVW_i^V$. The matrices $Q$, $K$ \& $V$ $\in \mathbb{R}^{N \times d_m}$ are the input matrices consisting of $N$ vectors of size $\mathbb{R}^{d_m}$, where $N$ is the input sequence length. $W_i^Q$, $W_i^K$ and $W_i^V \in \mathbb{R}^{d_m \times d_h}$ are linear transformations. The matrix $A_i$ is known as the attention matrix and the softmax is taken row-wise in the product $QW_i^Q(KW_i^K)^T$. Self-attention is expensive due to the matrix $A_i$ being of size $O(N^2)$. 

\textbf{Sparse Attention} To alleviate the quadratic computation in self-attention. Researchers have proposed a variety of \textit{sparsification} techniques to reduce the size and memory of computing $A_i$. These techniques compute some subset of the values of $A_i$ controlled by a mask matrix $M$, reducing the runtime of MHSA \cite{longformer, sparse-transformer} by computing: 
\begin{equation} \label{eqn:sparse-attention}
    \underbrace{\underbrace{[\emph{softmax}(\overbrace{M \otimes QW_i^Q(KW_i^K)^T)]}^{\sddmm{}}}_{A^s_i}VW_i^V}_{\spmm{}}
\end{equation} where mask $M$ is a mask of 0s and 1s and $\otimes$ is a pair-wise product. The product: $M \otimes (QW_i^Q(KW_i^K)^T)$ in traditional sparse computing terminology is a sampled dense dense matrix multiplications (SDDMM), whilst the product: $A^s_iVW_i^V$ is a sparse matrix dense matrix multiplication (SpMM). However, compared to the sparsity levels studied in sparse computing literature, $M$ is both moderately sparse and regular. Hence we term these operations appropriately prefixed with $regular$ as \sddmm{} and \spmm{} respectively. We define regularity in section \ref{Overview}.

\textbf{Sparse-MHSA Patterns} A variety of sparse transformers have been proposed in the literature \cite{longformer, longshort, sparse-transformer, reformer}. For example, the strided and windowed pattern (figure \ref{fig:regular-sparsity} left and middle) which implement Longformer (written in TVM) \cite{longformer, tvm}, and the blocked pattern (figure \ref{fig:regular-sparsity} right) which implements Reformer (written in JAX), and sparse-transformer (written in \triton{}) \cite{triton, jax, sparse-transformer}. SPLAT generates high-performance GPU code for all the patterns, avoiding the cumbersome need to hand-write complicated kernels.  

\begin{wrapfigure}[11]{R}{0.5\textwidth}
%\begin{figure}[H]
    \includegraphics[width=0.48\textwidth]{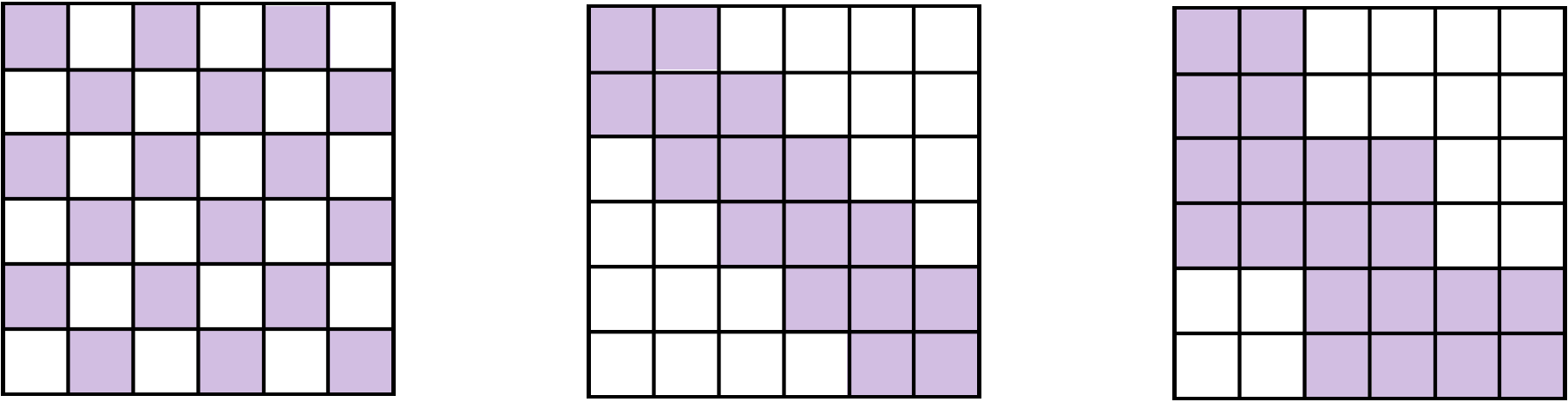}
    \caption{Examples of 3 commonly occurring sparse-MHSA patterns in the literature. Strided (left figure), Windowed (middle figure) \cite{longformer}, Blocked (right figure) \cite{sparse-transformer, reformer}. Full attention computes all points}
    \label{fig:regular-sparsity}
%\end{figure}
\end{wrapfigure}

\textbf{Sparse Formats} To obtain memory savings and performance benefits, sparse-kernels operate on data structures that only store the non-zero values in a sparse matrix. These data structures consist of non-zero values and their respective metadata. The metadata indicates the index of the trailing and leading dimension of a non-zero value. For example, the compressed-sparse-row (CSR) representation in figure \ref{fig:example-csr-triton} (b) contains the \emph{rowPtr} and \emph{colInd} arrays, indicating the leading and trailing dimensions of non-zero values in the \emph{values} array. Many sparse formats have been proposed in the literature including: COO, CSC, BCSR, ELLPACK, DIA, CSF \cite{sparse-formats,spmm-survey}, to name a few. We introduce a custom format: affine-compressed-sparse-row (\ds{}) (see section \ref{acsr}) that leverages the regularity in sparse-MHSA for performance benefits and memory savings.  

\textbf{Graphics Processing Units} GPUs are programmable single-instruction-multiple-thread machines. GPU kernels comprise of threads, with multiple threads constituting a warp, which operates in SIMD lock-step fashion. Multiple warps constitute a thread-block, which shares programmable L1 caches. Three important factors that reduce the performance of GPU kernels. (1) Thread-divergence: when threads within a warp have different control flow. (2) Un-coalesced memory requests: when consecutive threads within a warp do not access contiguous memory locations. (3) Low L1 cache re-use: when threads within a block do not reuse values in programmable L1 caches. For a greater understanding of GPU architectures see \cite{cuda-book}.

\newpage
\section{Motivation}
\label{Motivation}
\begin{figure}[t]
    \centering
    \includegraphics[width=\linewidth]{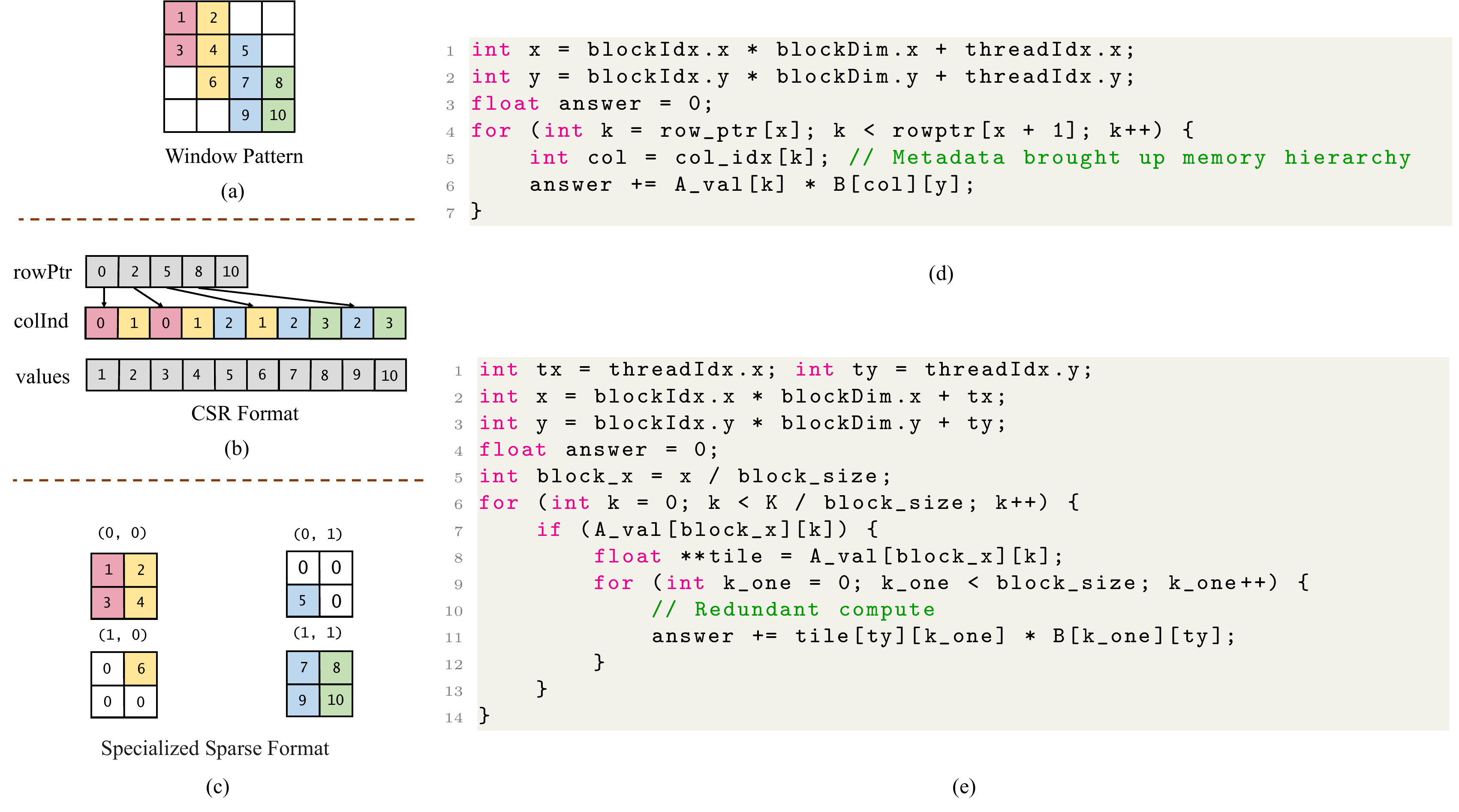}
    \caption{A comparison between SpMM implementations that use the CSR format (b), and a specialized format (c). (d) and (e) are naive SpMM implementations of $C=AB$, when $A$ is represented as a CSR and the specialized format of (c), respectively.}
    \label{fig:example-csr-triton}
\end{figure}

\begin{wrapfigure}[18]{r}{0.5\textwidth}
%\vspace{-25pt}
    \includegraphics[width=0.48\textwidth]{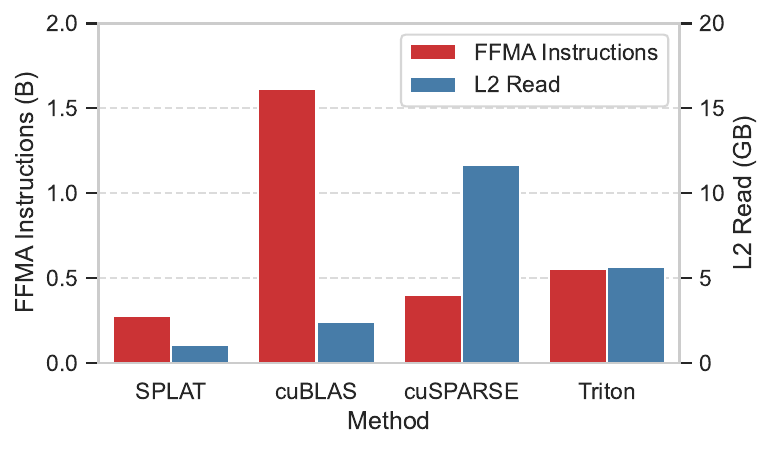}
    \caption{Profile of \spmm{} sparse-primitive implemented in SPLAT, \cublas{}, \cusparse{} and \triton{}. Matrices are 1024x1024 with sparse matrices in the window format (see figure \ref{fig:regular-sparsity} - middle) at 24\% density. FFMA is an FP32 fused multiply-add instruction and L2 read is the amount of data-traffic (in GB) from L2 to L1 cache. Lower is better.}
    \label{fig:motivational-study}
\end{wrapfigure}

%% We remove this paragraph due to repition.
%The performance characteristics of sparse kernels are coupled with the data formats they operate on. Different sparse formats store metadata, representing the coordinates of non-zero values, at varying levels of granularity with distinct memory layouts. On the one hand, general sparse formats (GSFs) such as the coordinate-sparse-row (CSR), and coordinate (COO) \cite{sparse-formats} are designed to store randomly sparse inputs and accordingly store metadata for \textit{each} non-zero value. On the other hand, custom sparse formats (CSFs) such as the blocked-compressed-sparse-row (BCSR) and diagonal (DIA) \cite{sparse-formats} are designed to store specific sparsity patterns and can leverage this structure to store metadata for a \textit{set} of non-zero values (e.g. blocks for BCSR, or diagonals for DIA). Both design choices place limitations on the thread-access patterns and optimizations applicable when accessing sparse data, influencing the performance of general libraries and hand-written kernels.  

Since the performance characteristics of sparse kernels are coupled with the data formats they operate on, we investigate the implications of using \textit{general sparse formats} (GSFs) and \textit{custom sparse formats} (CSFs) in the moderately sparse context of sparse-MHSA (with 10-50\% of the non-zero values computed). Consider the SpMM kernel, $C=AB$ where $A$ is a sparse tensor used in the last step of sparse-MHSA. When $A$ contains the windowed non-zero pattern (see figure \ref{fig:example-csr-triton} (a)) at a density level of 24\%, GSFs and CSFs manifest sub-optimal memory and compute profiles when used in either general sparse libraries or hand-written kernels. 

\textbf{General sparse formats lack performance.} GSFs are designed to store random \& extremely sparse patterns (with <1\% of the values computed) by storing metadata for \textit{each} non-zero value, occupying \emph{O(f(nnzs))} space. However, their metadata storage results in considerable data moved through the GPU memory hierarchy and results in sub-optimal memory profiles and performance.
%% I don't like this sentence. Metadata overhead is a vague and loaded term.
%In the sparse-MHSA setting, the high metadata overhead results in sub-optimal memory profiles. 
Consider the naive SpMM implementation operating on a CSR in figure \ref{fig:example-csr-triton} (d). When reading a value from sparse matrix \emph{A}, its respective column index must be read (line 5) to multiply with the correct row from \emph{B} which results in 3 loads of (1) \emph{col\_idx}, (2) $A$, and (3) $B$ to L1-caches and register files. However, the sparsity levels of sparse-MHSA layers are moderate, with up to 2048 megabytes of data produced per layer 
%\charith{very specific, either give an example or just say a vague term. Good to give an example. Also the figure is for what? Is it for longformer? which pattern? How large are these matrices??} 
in SOTA sparse architectures \cite{mistral, gemma-two}, resulting in non-zero data (from $A$ and $B$) and metadata (from \emph{col\_idx}) contending for space in caches and register files. Moreover, this contention occurs within every iteration of the inner loop (lines 4-6) resulting in frequent cache evictions and extraneous data moved from L2 to L1. This is not mitigated even in heavily optimized vendor-libraries for sparse computations like \cusparse{}, which operate on CSRs. As seen in figure \ref{fig:motivational-study}, \cusparse{}'s SpMM transfers 4.73x more data from L2 to L1 compared to \cublas{}'s dense matrix-multiplication, despite executing 1/3rd of the compute instructions.

\textbf{Custom sparse formats lack generality.} CSFs are designed to store specific sparsity patterns, reducing the metadata storage required to represent the coordinates of non-zero values. However, using these formats to store sparsity patterns that these formats were not designed for results in redundant compute and storage. Consider the naive SpMM implementation operating on the specialized format that stores block-like sparsity patterns in figure \ref{fig:example-csr-triton} (c). When adopting the same format to represent the window pattern, boundary conditions result in redundant storage of 0s in tiles \emph{(0,1)} and \emph{(1,0)}, incurring redundant compute in line 11 within the inner loop of lines 6-14. Moreover, more data is read than is necessary, resulting in extraneous traffic through the memory hierarchy. This is not mitigated even in highly optimized hand-written kernels. As seen in figure \ref{fig:motivational-study}, the highly optimized block-sparse kernels \cite{block-sparse}, hand-written kernels written in triton that operate on a similar CSF, execute 1.4x the floating point operations compared to \cusparse{}.

\textbf{SPLAT.} In this work, we recognize the issues with GSFs and CSFs in the moderately sparse regime of sparse-MHSA and aim to bridge this gap by introducing a new sparse-format: affine-compressed sparse-row (\ds{}). Core to its design is our realization that the geometry of sparse-MHSA patterns is regular, and this regularity only needs metadata in $O(rows)$ as opposed to $O(nnzs)$ like in GSFs, without compromising generality like in CSFs. In doing so, our code-generation mechanism, SPLAT, produces high-performance implementations of sparse-MHSA by reducing the number of compute instructions and data traffic across the memory hierarchy as we observe in figure \ref{fig:motivational-study}.

\section{Overview}
\label{Overview}
\textbf{SPLAT's Workflow} Figure \ref{fig:system-overview} shows the workflow of SPLAT. SPLAT takes an input mask and code-generates high-performance, sparse-MHSA implementations just-in-time. Its code-generation strategy proceeds in three phases. 

First, it proceeds with two analysis passes. The first pass analyzes the input mask and ensures that pre-conditions are met for the correctness of later code-generation passes. The second pass generates information required for certain optimisations later code-generation passes can exploit. Second, it proceeds with 3 kernel code-generation passes, producing the \sddmm{} (section \ref{section:sddmm}), Softmax, and \spmm{} (see section \ref{spmm}) kernels used to implement sparse-MHSA. Third, it proceeds with an end-to-end code-generation pass (see section \ref{splat}) that allocates the necessary memory and creates auxiliary objects required for the correctness of kernel optimizations. The output of the end-to-end code-generation phase is a compiled function that can be used in transformer models to implement the sparse-MHSA mechanism. Our code-generation scheme produces high-performance sparse-MHSA implementations that store sparsity in our novel custom format: affine-compressed-sparse-row - \ds{} (see section \ref{acsr}) that leverages the regularity of these patterns. 

\subsection{Affine-Compressibility and Regularity}

As observed in section \ref{Motivation} appropriate sparse formats for sparse-MHSA kernels should \textit{compress} metadata, reducing the number of bytes used to store the indices of each non-zero value. Moreover, such a compression scheme should be able to precisely represent non-zero values' metadata across a variety of sparse-MHSA patterns without redundancy. We achieve both by observing that the \textit{point-set} of commonly occurring sparse-MHSA structures (like in figure \ref{fig:regular-sparsity}), consists of rows that are \textit{affine-compressible}, and are therefore \textit{regular}. This observation enables us to create a novel sparse-format that \textit{symbolically} stores the metadata for each row of a regularly sparse structure through an affine function. 

\textbf{Point-sets} To analyze the geometric properties of sparse-MHSA structures, we interpose their input-masks onto the cartesian coordinate system. For a mask, $M$, consisting of 0s and 1s, we map the point $M[i][j]$ to the point $(j,i)$. We define the point-set of an input-mask as the set of all points that are 1, i.e. the set of all $(j,i)$ such that $M[j][i] = 1$.

\textbf{Affine-Compressibility} Affine-compressibility is a property of sets of points on the cartesian coordinate system. It states that a set of points can be compressed, such that they consecutively neighbor each other along the x-dimension. \begin{definition}
    Consider a set of points: $P = \{(x_1, y_1), (x_2, y_2), ..., (x_k, y_k)\}$ on the coordinate system. $P$ is affine-compressible if and only if: $$\exists a, b \in \mathbb{R}, \text{such that,} \forall i \in [k-1] \text{, } x_i \cdot a + b + 1 = x_{i+1} \cdot a +b $$ We denote $a, b$ as the affine-indices of $P$.
\end{definition}
For example, the set $P_1=\{(0, 0), (2, 0), (4, 0), (6, 0)\}$ is affine-compressible with affine-indices: $a=0.5, b=0$, however the set $P_2=\{(0,0), (2, 0), (4, 0), (5, 0)\}$ is not.

\textbf{Regularity} Regularity is a property of a sparse-MHSA mask, building upon the concept of affine-compressiblity. We define a sparse-MHSA mask, $M$, to be \textit{regular} iff every row in its corresponding point-set, $P$, is affine-compressible.
\noindent Hence, a regularly sparse mask is amenable to metadata compression by symbolically storing the dense indices of the trailing-dimension of a sparse matrix. For example, in figure \ref{fig:memory-waterfall} (b), we see for each row of the window pattern the respective linear-transformation (denoted as $a$) and translation (denoted as $b$). 

\begin{figure}[t]
    \centering
    \includegraphics[width=\linewidth]{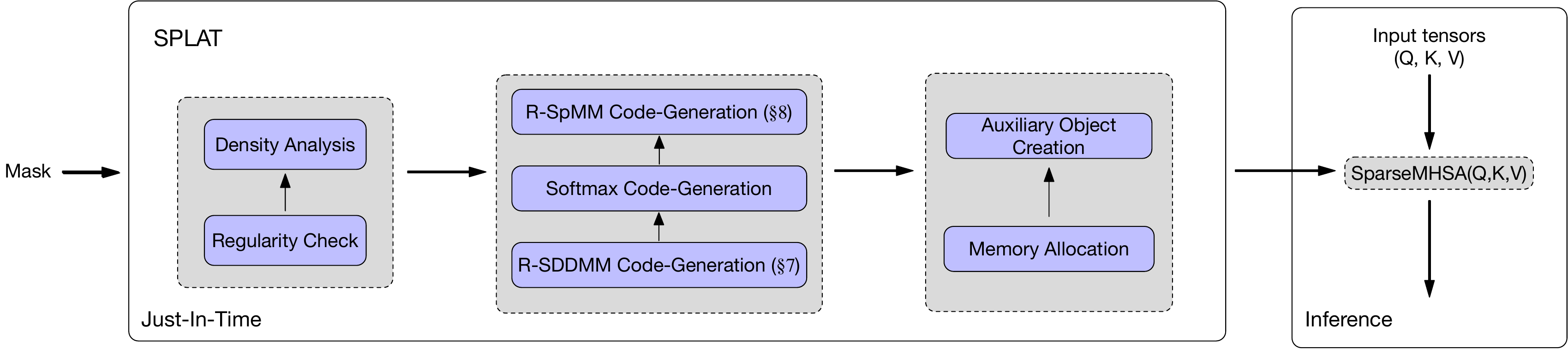}
    \caption{An overview on SPLAT's inner mechanics and how its just-in-time strategy produces compiled sparse-MHSA kernels for inference.}
    \label{fig:system-overview}
\end{figure}

\section{Affine-Compressed Sparse-Row}
\label{acsr}
We introduce the \ds{} format to store regularly sparse matrices. The \ds{} format leverages the regularity of sparse-MHSA matrices to store metadata in the order of number of rows with its metadata, the affine-indices, exposing various optimization opportunities. We detail the construction of the ACSR in section \ref{construction}, and the optimizations its metadata exposes in section \ref{properties}.

\subsection{\ds{} construction}
\label{construction}

The \ds{} comprises of two arrays: non-zero values, and metadata. The metadata symbolically records the index of the trailing dimension for each non-zero value in a particular row. \ds{} represents a sparse matrix by computing the affine-indices per row and compressing data across the trailing dimension such that non-zero values consecutively neighbor each other. For example, in figure \ref{fig:memory-waterfall}, the original 2-D matrix in figure \ref{fig:memory-waterfall} (a) is compressed across the trailing-dimension to \ref{fig:memory-waterfall} (b). Each row in \ref{fig:memory-waterfall} (b) has the triplet: $a$ (linear-transformation), $b$ (translation), and $nnzs$ (number of non-zero-values) as metadata. If $sparse_i$ is the index of a non-zero value's trailing dimension in the \ds{}, then $(sparse_i - b)/a$ is the index of the trailing dimension in the original sparse matrix. Consider the location with value 14 in figure \ref{fig:memory-waterfall} (b); it is at $sparse_i = 1$ and has $a=1$, $b=-1$, and $nnzs=4$ as metadata. The index of its trailing dimension in figure \ref{fig:memory-waterfall} (a) is thus $\frac{(1 + 1)}{1} = 2$. The metadata consists of the ($a$, $b$, $nnzs$) triplet per row, occupying $O(rows)$ rather than $O(nnzs)$ space. 

However, to construct an ACSR, the affine-indices for each row of a sparse 2-D matrix need to be computed. This is error-prone to implement and can be avoided by generating the affine-indices by solving a set of linear equations. Consider the $y^{th}$ row in a regularly-sparse 2-D matrix whose first two points are at column indices: $i_{0,y}$ \& $i_{1,y}$. Then, solving: $$\mathtt{Solve}\bigg(\begin{bmatrix} i_{0,y} & 1\\ i_{1,y} & 1 \end{bmatrix} \begin{bmatrix} a\\ b \end{bmatrix} = \begin{bmatrix} 0 \\ 1\end{bmatrix}\bigg)$$ 
Uncovers the affine-indices for that particular row. After computing the affine-indices for each row, we can check to see if the entire pattern is then affine-compressible. We do this by computing: $i_{x,y}\times a_{y} + b_{y} = i_{x-1,y}\times a_{y} + b_{y} + 1$, where $a_{y}, b_{y}$ are the affine-indices for the $y^{th}$ row and $i_{x,y}$ is the column index of the $x^{th}$ non-zero value in row $y$. This ensures that we do not incorrectly operate on masks that are not regular. 

%\charith{what if there are deviations in the middle? you need a check and say in practice this has been enough. Also, you never mention what are the metadata. You need to better tie concepts.}

\begin{wrapfigure}{R}{0.5\textwidth}
    \includegraphics[width=0.48\textwidth]{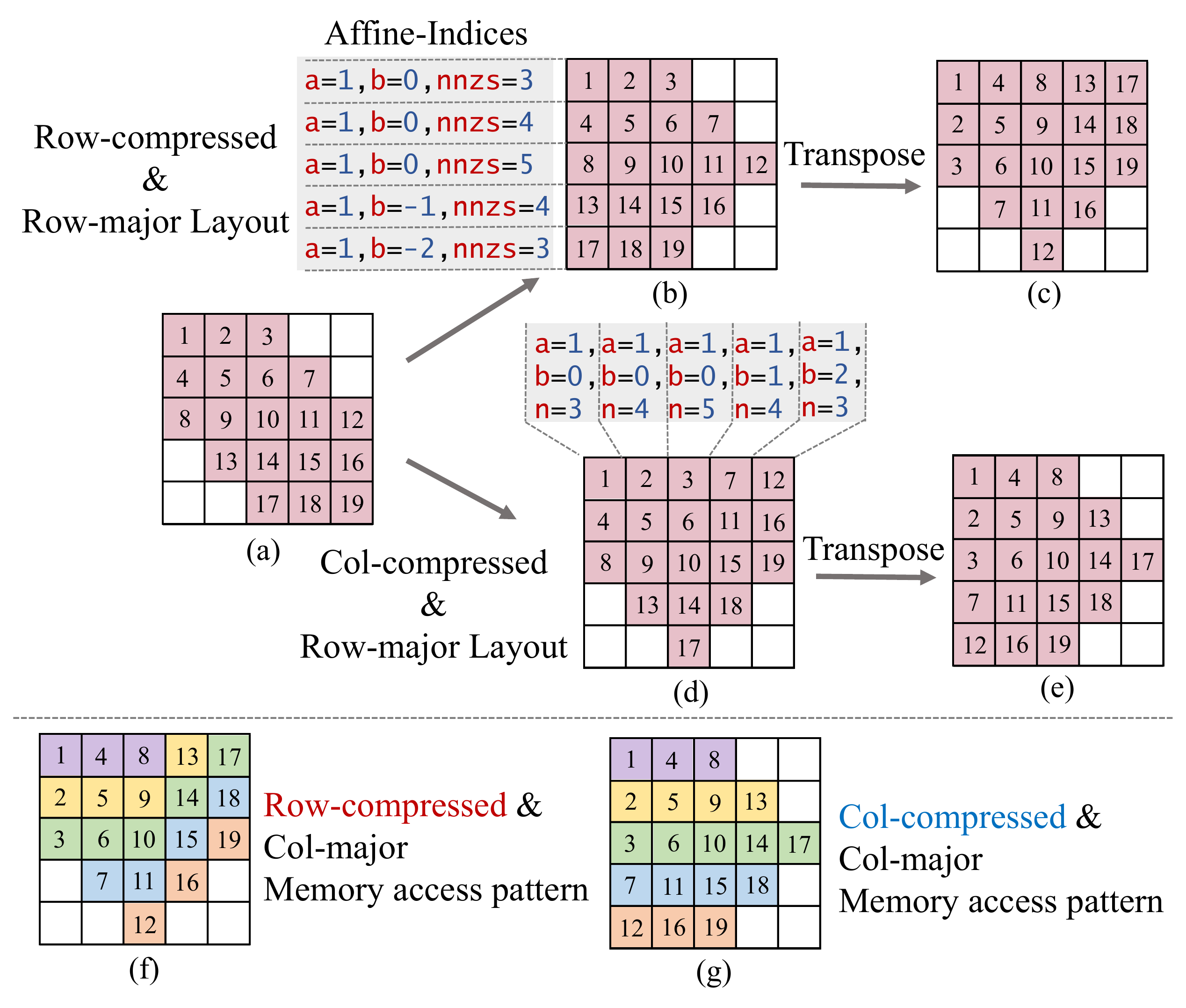}
    \caption{The 4 different data-layouts an \ds{} can take: (b) row-wise compressed row-major, (c) row-wise compressed col-major, (d) col-wise compressed row-major, (e) col-wise compressed col-major. For (f) and (g), colors represent elements of the same column. The \emph{a}, \emph{b}, and \emph{nnzs} represent a row's linear-transformation, translation, and number of non-zero values respectively, constituting the metadata. The \emph{a} and \emph{b} variables are the affine-indices of a row.}
    \label{fig:memory-waterfall}
\end{wrapfigure}

\subsection{\ds{} Properties}
\label{properties}

The \ds{} exposes novel optimization opportunities for \sddmm{} and \spmm{} kernels at the moderate sparsity levels observed in sparse-MHSA.

\textbf{Reduction in Predicated Execution.} Operating on sparse inputs may result in an imbalance of work across threads within a warp as certain input regions are potentially more dense than others. The ACSR, in storing the dense indices of the trailing dimension symbolically via affine-indices, exposes which regions of a sparse tensor are identical at the granularity of a row. For example, if different rows have identical linear-transformations, translations, and number of non-zero values, then they have data placed in identical trailing indices. Rows with identical affine-indices can be re-mapped to operate on threads within a warp to reduce predicated execution. 

\textbf{Favorable read/write access patterns.} In \spmm{} kernels, memory accesses to conventional sparse-formats can be un-coalesced. To coalesce these accesses, contiguous elements in a column need to be laid out in contiguous memory addresses, which our construction in section \ref{construction} does not do. Fortunately, regularly sparse kernels are symmetric and are therefore also \textit{column-wise} affine-compressible. When compressing data across the trailing dimension and laying out data in column-major, contiguous memory addresses in the \ds{} contain contiguous elements in a column of a sparse matrix as shown in figure \ref{fig:memory-waterfall} (g). 

\textbf{Fast indexing.} Certain sparse kernels check whether an index in a sparse 2-D matrix is non-zero by traversing a region of values in a sparse-format. For example, to identify if a point $(dense_i, dense_j)$ (leading, trailing dimensions resp.) is non-zero in a CSR requires a traversal of all the points $[rowPtr[dense_i], rowPtr[dense_i + 1]]$. However, the ACSR can compute the answer in $O(1)$ time and metadata accesses by computing: $dense_j\%(1/a)==0 \wedge dense_j +b > 0$, where $a$ and $b$ are the affine-indices of row $dense_i$. We exploit this in our \spmm{} and \sddmm{} kernels.

\section{Regularly Sparse Primitives}
\label{Regular Sparsity}
Sparse-MHSA is implemented by 3 kernels: \sddmm{}, Softmax, and \spmm{}. First, the \sddmm{} kernel computes $M\otimes(QW^Q)(KW^Q)^T$. Both the $QW^Q$ and $KW^K$ kernels are dense, producing an output \ds{}. Second, this output is ingested by a softmax kernel, computing the softmax for each input row, producing an \ds{}. Last, the \spmm{} consumes the output of the sofmax kernel: $A^s_i$, computing the $A_i^sV_w$ product, where $A^s_i$ is stored as an \ds{}. 

However, with the introduction of the \ds{} format, novel optimizations and thread-access patterns to reference data need to be explored due to its unique metadata structure. These optimizations and access patterns should ensure that threads re-use data in L1 caches and exhibit spatial locality, issue coalesced memory writes and reads, and have identical control flow. Producing access patterns that achieve all 3 traits in the context of sparse-MHSA is challenging due to the geometric diversity of input masks. We propose a unified set of optimizations and a code-generation mechanism that produces high-performance code for \spmm{} and \sddmm{} kernels that can carefully balance these 3 traits. We additionally propose an end-to-end code-generation scheme that leverages these kernels to implement sparse-MHSA. Section \ref{section:sddmm} and \ref{spmm} details the individual code-generation and optimization opportunities for \sddmm{} and \spmm{} kernels, and section \ref{splat} details the end-to-end code generation scheme for sparse-MHSA.

\section{High-Performance \sddmm{}}
\label{section:sddmm}
An important optimization frequently applied to GPU implementations of SDDMM kernels is tiling, which improves reuse and reduces thread-divergence. This involves deciding a mapping of thread-blocks to outputs. Different tiling strategies have been explored within the context of random and extreme sparsity. \cite{aspt} re-orders the input to enhance reuse of tiles, \cite{sparse-gpu-kernels-dl} proposes a 1-dimensional tiling scheme tied to the CSR, and \cite{sddmm-inspector-executor} uses a model-driven inspector-executor approach to increase reuse.  

Such strategies are either targeted towards extreme sparsity levels where the cost of inspection and re-ordering is low, or are tied to particular sparse formats. Therefore, they are expensive to apply at the moderate sparsity levels in sparse-MHSA with the unique metadata layout of the \ds{}. Comparatively, we leverage the regular nature of the sparsity patterns and the ACSR format to provide a novel, inexpensive tiling strategy for the \sddmm{} kernel (see section \ref{section:greedy-tiling}). We show that our tiling approach increases cache reuse, and memory coalescing whilst reducing thread-divergence and redundant compute with \textit{strong optimality guarantees} (see section~\ref{section:greedy-tiling}). 

\subsection{Observations} 
\label{sddmm:observations} 

The geometric diversity of sparse-MHSA patterns gives rise to many possible arrangements of thread-blocks over the output, $C$. Each arrangement trades off different factors that impact performance. We categorize each sparse pattern as either polygonal or strided. Polygonal patterns comprise of non-zero values that are clustered together, with no gaps between them like the windowed and blocked pattern. Strided patterns consist of non-zero values which have constant gaps between them, each non-zero value having no neighbor. 

\label{section:case-study-sddmm-tiling}
\begin{figure*}[t]
    \centering
    \includegraphics[width=\textwidth]{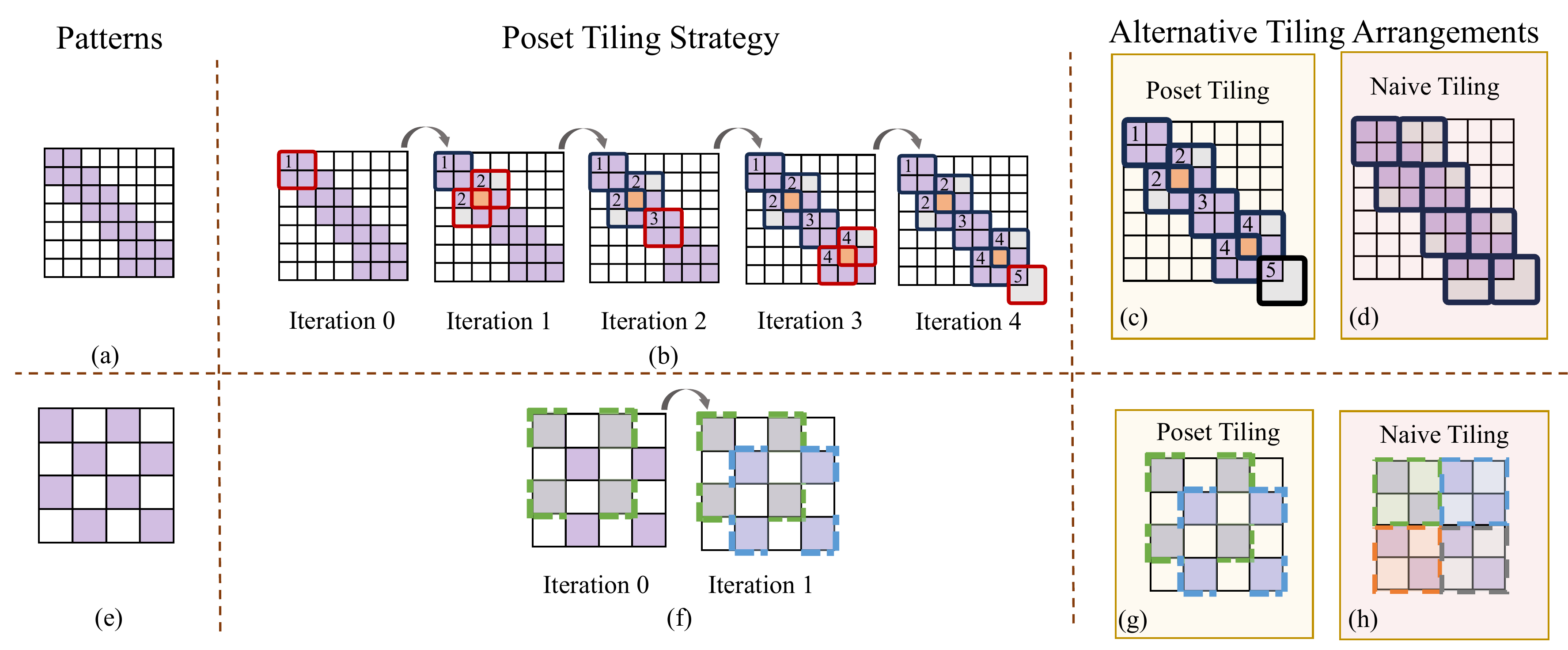}
    \caption{Different ways thread-blocks can tile strided and polygonal patterns. (a) and (e) are the two patterns. (b) and (f) demonstrate our novel poset tiling strategy. (c) and (d) show two strategies to tile polygonal patterns. Gray represents thread-divergence, and orange represents redundant compute. (g) and (h) show two strategies to tile strided patterns. Numbers on the mask represent what iteration in the for-loop (line 4 of algorithm \ref{alg:greedy-tile}) the thread-block was placed.}
    \label{fig:strided-different-tile}
    \vspace{-15pt}
\end{figure*}

\textbf{Polygonal Patterns.} Figures \ref{fig:strided-different-tile} (c) and \ref{fig:strided-different-tile} (d) show two valid tiling arrangements for the same polygonal pattern. \ref{fig:strided-different-tile} (d) incurs more threads with divergent control-flow compared to \ref{fig:strided-different-tile} (c), as more threads within a thread-block exceed the boundary of the pattern and are predicated to terminate, diverging from the threads that compute output values. Instead \ref{fig:strided-different-tile} (c) incurs threads with redundant compute as thread-blocks that overlap (orange points) compute the same values. The more performant tiling arrangement between the two will depend on the relative costs associated with thread-divergence, redundant compute, and number of thread-blocks.

\textbf{Strided Patterns.} Figures \ref{fig:strided-different-tile} (g) and \ref{fig:strided-different-tile} (h) show two valid tiling arrangements for the same strided pattern. \ref{fig:strided-different-tile} (h) exhibits low spatial locality compared to \ref{fig:strided-different-tile} (g), as thread-blocks operate on outputs that do not re-use rows and columns from the input. Instead, thread-blocks in \ref{fig:strided-different-tile} (g) issue un-coalesced reads to input matrices as they operate on outputs with a constant stride. Additionally, \ref{fig:strided-different-tile} (h) exhibits increased divergent control-flow and uses more thread-blocks compared to  \ref{fig:strided-different-tile} (g). The more performant tiling arrangement between the two will depend on the relative costs associated with un-coalesced memory accesses, divergent control flow, spatial locality, and number of thread-blocks. 

Our observations indicate that 4 factors impact the performance of a tiling arrangement. (1) The amount of redundant compute between thread-blocks that overlap and compute the same output. (2) The amount of thread-divergence within a thread-block by being placed on irregular boundary conditions. (3) The amount of reuse within a thread-block by computing outputs that share either rows/columns of input matrices $A$ and $B$. (4) The number of memory access/write requests issued by a warp that are coalesced by reading/writing to contiguous memory locations. A good code-generation scheme should generate a tiling arrangement that reduces the cost of each factor. 

\subsection{\sddmm{} Performance Characterisation} 
\label{sddmm:characterisation}

We first develop a cost model that explicitly reasons about each of the four factors that affect the performance of a tiling strategy. We achieve this by developing expressions to compute each of these factors as a function of the thread-block arrangement of a given tiling strategy. 

%There is an exponential number of tiling arrangements for a particular sparse pattern. If we have a point-set, $P$, that we are covering with thread-blocks of size $m\times n$, then there are at least $|P| \choose |P|/mn$ tiling arrangments to choose from. Enumerating each becomes intractable for moderately sparse point-sets.

%Consider the 4 arrangements in figure \ref{fig:strided-different-tile}. In iteration 1, the two thread-blocks placed each incur some redundant compute by partially overlapping each other, and some thread-divergence by exceeding the boundary conditions of the mask. Similarly, in \ref{fig:strided-different-tile} (b), each thread-block experiences thread-divergence and low re-use by being placed on irregular boundary conditions. The definitions of redundant compute, thread-divergence, and reuse are thus coupled to a thread-block's placement over a mask. 

\textbf{Thread-Blocks.} We define a thread-block to be a mapping between a logical rectangle of threads of size $m \times n$ to points on a mask, $M$.

\begin{definition}
A thread-block $TB_k$ consisting of $m \times n$ threads that partially covers a point-set $P$ is defined by a tuple: $(t, s), t\in \mathbb{N}\times\mathbb{N}, s \in \mathbb{N}$. We further define the compute of a thread-block as $Comp(TB_k) = \{t + (i*s, j*s) | i\in \{0,\dots, m-1\}, j \in \{0, \dots, n-1\}\}$.

Finally, we define its cover, anchor-point, and stretch factor as:
\begin{align*}
    Cov(TB_k) = Comp(TB_k) \cap P, \text{ } Anc(TB_k) = t, \text{ } Str(TB_k) = s
\end{align*}
\end{definition}

The cover and anchor-point represent the points a thread-block computes and its top-left corner (its translation from the origin) respectively. For example, in figure~\ref{fig:strided-different-tile} (g), the cover of the two thread-blocks is $Cov(TB_0) = \{(0,0), (0, 2), (2, 0), (2, 2)\}$ (green thread-block), with $Anc(TB_0) = (0, 0)$, and $Cov(TB_1) = \{(1, 1), (3, 1), (1, 3), (3, 3)\}$ (blue thread-block), with $Anc(TB_1) = (1, 1)$. The stretch factor of a thread-block determines how far apart threads in neighboring rows and columns will be placed when covering a point-set. For example, the two thread-blocks in figure \ref{fig:strided-different-tile} (g) have a stretch factor of 2. 

\subsubsection{Factors affecting performance} Definition \ref{thread-block-definition} gives the mathematical formulation of the four factors impacting the performance of a \sddmm{} kernel. We give intuitions for those definitions next. Note that thread-divergence and redundant compute are aggregate sums, while reuse and memory coalesced requests are averages across all thread-blocks.

\textbf{Thread-divergence.} Within a thread-block, threads that exceed the boundary conditions of a mask deviate control flow from threads that do not. Although thread-divergence happens within a warp, due to the irregular boundary conditions in regularly sparse masks, oftentimes threads that exceed the boundary of a mask exhibit thread-divergence. Therefore, we define the collective thread-divergence of an arrangement as the number of threads that do not cover a point in the point-set. (See $\phi_{TD}$ in definition \ref{thread-block-definition}).

\textbf{Redundant Compute} An arrangement's redundant compute is the number of excess threads that do not do useful work across all thread-blocks. This amounts to a sum of all the threads in the arrangement subtracted by both the number of points in the point-set and the number of threads that have divergent control flow. (See $\phi_{R}$ in definition \ref{thread-block-definition}).

\textbf{Reuse of Thread-block.} Threads within a thread-block that do useful work usually reuse values of the input rows or columns. The threads that do not do useful work fall into two categories: (1) Divergent threads, (2) redundant threads. To compute the reuse of a thread-block, we compute the fraction of threads within a thread-block that are both not divergent and redundant. Hence, we define the reuse of an arrangement to be the average reuse across all thread-blocks. (See $\phi_{RU}$ in \ref{thread-block-definition}).

\textbf{Degree memory requests are coalesced.} The degree to which memory requests of a warp within a thread-block are coalesced is \textit{inversely proportional} to a thread-block's stretch factor. Since both the inputs are dense, the larger the stretch factor, the larger the stride in reads issued to inputs, and writes issued to outputs. Hence, we define the amount of memory coalescing as the average stretch factor across all the thread-blocks in an arrangement. (See $\phi_{CMR}$ in \ref{thread-block-definition}).

\begin{definition}\label{thread-block-definition}
Consider an arrangement of thread-blocks, $TB=\{TB_1, TB_2, ... TB_\lambda\}$ each containing $m \times n$ threads, covering a point-set, $P$ such that $\bigcup_{TB_i \in TB}Cov(TB_i) = P$. We define its collective thread-divergence ($\phi_{TD}$), redundant compute ($\phi_R$), reuse ($\phi_{RU}$), and coalesced memory-requests ($\phi_{CMR}$):
\begin{align*}
&\phi_{TD} = \bigg| \bigg( \bigcup_{TB_i \in TB} Comp(TB_i) \bigg) \backslash P \bigg| & \phi_R = \lambda m n - |P| - \phi_{TD} & \\
&\phi_{RU} = \frac{mn - \frac{\phi_{TD}}{\lambda} - \frac{\phi_{R}}{\lambda}}{mn} = \frac{|P|}{\lambda m n} & \phi_{CMR} = \frac{1}{\lambda} \sum_{TB_i \in TB} \frac{1}{Str(TB_i)} &
\end{align*}
\end{definition}

\noindent We illustrate divergent threads and redundant threads in figure \ref{fig:strided-different-tile} (c) as gray and orange respectively, with $\phi_{TD} = |\{(1,3), (3,1), (6,4), (6,7), (7,7), (7,6)\}| = 6$, and $\phi_{R} = 7 * 4 - 20 - 6 = 2$. We compute the reuse of \ref{fig:strided-different-tile} (g): $\frac{2\times2-0-0}{2\times2} = 1$ and \ref{fig:strided-different-tile} (h): $\frac{2\times2-8/4 - 0}{2\times2} = \frac{1}{2}$, as well as the coalesced memory requests of \ref{fig:strided-different-tile} (g): $\frac{1}{2}(\frac{1}{2} + \frac{1}{2}) = \frac{1}{2}$, and \ref{fig:strided-different-tile} (h): 1.

\textbf{Cost Model.} A performant tiling arrangement will minimize redundant compute, thread-divergence, and stretch-factors of thread-blocks whilst maximizing reuse and coalesced memory requests. This amounts to minimizing $\phi_{TD}$, and $\phi_R$, whilst maximizing $\phi_{RU}$ and $\phi_{CMR}$.

On one hand, minimizing both $\phi_{TD}$ and $\phi_{R}$ and maximising $\phi_{RU}$ corresponds to reducing $\lambda$, since the dimensions of thread-block: $m$, $n$ and the size-of the point-set: $P$, are all fixed. Therefore, optimal tiling arrangements that reduce the costs of thread-divergence, redundant-compute, and low reuse will minimize the number of thread-blocks used. On the other hand, maximising $\phi_{CMR}$ corresponds to reducing $Str(TB_i)$ for all thread-blocks in the arrangement. Therefore, a good cost function will increase with the number of thread-blocks used, and decrease when $\phi_{CMR}$ increases. 

\begin{definition}
\label{cost-function}
Consider an arrangement of thread-blocks, $TB=\{TB_1, TB_2, ... TB_\lambda\}$ each containing $m \times n$ threads covering a point-set, $P$. Its cost is denoted as $Cost(TB)$ and is computed as follows: $Cost(TB) = \frac{\lambda}{\phi_{CMR}}$
\end{definition}

\subsection{Poset Tiling} 
\label{section:greedy-tiling}
We develop a tiling strategy - poset tiling - to tile patterns with \textit{optimality} guarantees according to our cost model. Given a mask, $M$, whose point-set is $P$, it outputs an arrangement of thread-blocks that covers $P$ by computing the anchor-points where thread-blocks should be placed. It computes these anchor-points by successively computing a set $\top$, using the comes-before (CB) relation. 

%\ahan{Cut this, just say to aid in our discussion, we define $\top$ and CB.} Algorithm \ref{alg:greedy-tile} gives the tiling algorithm; it proceeds iteratively. During an iteration, it partially builds the answer by computing the anchor-points of thread-blocks that cover a large uncovered portion of $P$ by computing a set, termed $\top$, using the comes-before (CB) relation.

\begin{definition}
Suppose we have a point-set, $P$ that is partially tiled by $TB=\{TB_1, TB_2, ... TB_k\}$. Then given two points, $(x_1, y_1), (x_2, y_2) \in P$, we say that $x \leftarrow y$ (i.e. x CB y) iff $x_1 \leq x_2 \wedge y_1 \leq y_2$. \\ Moreover, let $P'$ be the set of un-covered points of $P$. We define $\top$ to be the set of points in $P'$ such that: $\forall p_i \in \top, \nexists p_j \in P'$ such that $p_j \leftarrow p_i$.
% We define $\top$ to be the set of points $\{p_1, p_2, ... p_k\}$ 
\end{definition}

\noindent Intuitively, the set $\top$ defined over a partially covered point-set, $P$, represents a set of uncovered points that, when used as the anchor-points of thread-blocks, cover a large uncovered portion of $P$. For example, in figure \ref{fig:strided-different-tile} (b) (Iteration 2) when the points in $\top = \{(4,5), (5, 4)\}$ are treated as the anchor-points of thread-blocks, we produce the arrangement in \ref{fig:strided-different-tile} (b) (Iteration 3). These additionally placed thread-blocks cover 6 uncovered points.

Algorithm \ref{alg:greedy-tile} demonstrates how poset tiling covers a point-set, $P$, using the set $\top$. It takes as input a point-set, $P$, to cover, and the dimensions of the thread-blocks ($m \times n$) used to cover $P$. It outputs a list of points representing the anchor-points of an arrangement of thread-blocks that covers $P$. It begins by initializing the answer ($AncPt$) to $\emptyset$ in line 1, and the set of uncovered points ($Rem$) to $P$ in line 2. It decides a suitable stretch factor to apply to thread-blocks (line 3) by calling a sub-routine $stretchFactorSelection$ (described in \ref{sddmm:stretch-factor-select}). The main loop in line 3 iterates until $Rem$ is $\emptyset$, which occurs when we have a cover of $P$. At each loop iteration, line 4 computes the current $\top$ ($\top_{curr}$) over the uncovered points $Rem$, adding this to the solution set $AncPt$. Finally, using $\top_{curr}$ as the anchor-points of thread-blocks stretched by a factor of $s$, lines 6-8 remove points from $Rem$ that will be covered. 

Figure \ref{fig:strided-different-tile} (b) Iterations 0-4 represent each iteration of the main loop in poset tiling. 

\begin{theorem}
Given point-set $P$ and arrangement of thread-blocks $TB_{poset} = \{TB_1, TB_2, ..., TB_{\lambda_{poset}}\}$, each of size $m \times n$, generated by algorithm \ref{alg:greedy-tile} to cover $P$. Let $TB_{opt} = \{TB_1, TB_2, ..., TB_{\lambda_{opt}}\}$ be the arrangement of lowest possible cost to cover $P$. Then, for the windowed and blocked pattern: \label{theorem:bound} \begin{equation}
    \frac{Cost(TB_{poset})}{Cost(TB_{opt})} \leq 1 + \frac{m}{l}
\end{equation}

Where $l$ is the maximum number of points in a row of the mask. Moreover, for the strided pattern, the cost of the arrangement is optimal. See Appendix C for the proof.  
\end{theorem}

\SetKwComment{Comment}{/*}{*/}
\SetKwInOut{Input}{Inputs}
\SetKwInOut{Output}{Output}
\begin{wrapfigure}{R}{0.5\textwidth}
    \begin{minipage}{0.5\textwidth}
        \begin{algorithm}[H]
        \caption{Poset Tiling}\label{alg:greedy-tile}
        %\begin{algorithmic}[1]
        \Input{$P, m, n$} 
        $AncPt \gets \emptyset$\;
        $Rem \gets P$\; %\Comment{$Rem$ are the points \textit{not} covered by TB}
        $s \gets stretchFactorSelection(P)$\; %\Comment{Stretch factor selection sub-routine}
        \While{$\cup_{TB_i \in TB}Cover(TB_i) \neq P$}{
            $\top_{curr} \gets f_\top(Rem)$\;
            $AncPt \gets AncPt \cup \top_{curr}$\;
            \For{$(x, y) \in \top_{curr}$}{
                $Rem \gets P - \{\bigcup_{i\in[m]}x+i \times s\} \times \{\bigcup_{j\in[n]}y+j \times s\}$\;
            }
        }
        \Output{AncPt}
        \end{algorithm}
    \end{minipage}
\end{wrapfigure}

\subsubsection{Stretch Factor Selection}\label{sddmm:stretch-factor-select}
The stretch factor of thread-blocks may influence the number of thread-blocks used in a tiling arrangement. For example, take the strided pattern in figure \ref{fig:strided-different-tile} (g) and \ref{fig:strided-different-tile} (h). By increasing the stride from 1 to 2, we go from the arrangement in \ref{fig:strided-different-tile} (h) which uses 4 thread-blocks, to the arrangement in \ref{fig:strided-different-tile} (a) which uses 2 thread-blocks. Therefore, a good stretch factor will balance the cost of issuing un-coalesced memory requests with the number of thread-blocks used in an arrangement. 

A naive stretch factor selection sub-routine will run poset tiling for every possible stretch factor, enumerating the cost of each arrangement and selecting the factor that corresponds to the lowest cost arrangement. Since the stretch-factor is bounded by the sequence length, $N$, this will return: $s = \argmin_{i \in [N]}\frac{\lambda^i}{\phi_{CMR}^i}$. Where $\lambda^i$ and $\phi_{CMR}^i$ are the number of thread-blocks and degree of coalesced memory-requests of an arrangement produced by poset tiling with stretch factor $i$, respectively. However, we make two observations that enable us to reduce the number of arrangements to search through. 

For polygonal patterns, stretching a thread-block will only reduce its cover (see Appendix A for the proof). Therefore, poset tiling will produce an arrangement of the lowest thread-block count when the stretch factor is 1. Hence, for polygonal patterns, we return 1.

However, for strided patterns, stretching a thread-block may increase its cover. To aid us in cutting down the space of stretch factors to search through, we make the following observation. Consider a strided pattern with mask, $M$. Define its stride to be the number of points between two successive non-zero values in a row of $M$, denoted as X. Then we have that the cost of an arrangement is minimized when applying algorithm \ref{alg:greedy-tile} with stretch factor $s = \argmin_{d_i \in factors(X)} \frac{\lambda^{d_i}}{\phi_{CMR}^{d_i}}$ (See Appendix B for the proof).

\newpage
\subsection{\sddmm{} Kernel Code-Generation}

\begin{wrapfigure}{R}{0.5\textwidth}
    \begin{minipage}{.47\textwidth}
        \input{code-listings/sddmm-kernel}
    \end{minipage}
\end{wrapfigure}

We show the code-generation pass and a naive \sddmm{} kernel in listing \ref{alg:sddmm-code-gen} and listing \ref{alg:sddmm-naive} respectively. The code-generation pass takes in an input-mask (line 1) and first applies poset tiling to generate the thread-block count, anchor-points, and stretch factor (line 5). It then uses the thread-block count to instantiate the \sddmm{} launcher and compiles this function (lines 7-8), returning the function pointer, func. Finally, it returns the function pointer and anchor-points (line 9).

\begin{wrapfigure}{R}{0.5\textwidth}
    \begin{minipage}{.47\textwidth}
    \input{code-listings/sddmm-code-gen}
    \end{minipage}
\end{wrapfigure}

We illustrate a naive implementation of the \sddmm{} kernel for brevity. The \sddmm{} kernel takes as input the anchor-points (idxToOut), stretch-factor (s), left (A), and right (B) input matrices, and space to store the output (C). Lines 7-8 use the thread-block id to index the map to recover the anchor-point of the thread-block the current thread belongs to. It uses this anchor-point to compute the row of A, $out_{iy}$, and column of B, $out_{ix}$, to dot-product in lines 10-12. Finally, it indexes the \ds{} metadata to get the affine-indices for the $out_{iy}$ row, applying the correct linear-transformation and translation to store the answer, out, at the correct index in the \ds{} non-zero values array, C (in lines 14-16). 

\section{High-Performance \spmm{}}
\label{spmm}
The R-SpMM kernel of sparse-MHSA consumes the softmax of the sparse output of the R-SDDMM kernel ($A_i^s$) and computes the $A_i^sV_w$ product, where $V_w$ is the dense matrix that results from the multiplication of $W_i^V$ with the value matrix $V$. $A_i^s$ is in the ACSR format and hence, similar to the \sddmm{} kernel, we need to devise novel techniques to generate high-performance \spmm{} code to leverage the full potential of the ACSR properties mentioned in Section~\ref{acsr}. We first show how to construct a naive R-SpMM kernel that uses the ACSR format, then incrementally present two optimizations that enable SPLAT to generate high-performance \spmm{} implementations.

\subsection{Observations}

The Algorithm in listing \ref{alg:spmm-naive} shows a naive implementation of the \spmm{} kernel $C=AB$, where $A$ is a sparse matrix represented in the \ds{} format and $B$ is a dense matrix. Traditionally, SpMM kernels iterate over the non-zero values of $A$ and multiply these with a corresponding value from $B$ (see figure \ref{fig:example-csr-triton} (d) for an example). However, in \spmm{} kernels, up to 50-70\% of $A$ can contain non-zero values. Rather than iterating over only the non-zero values, we treat the \spmm{} kernel as a dense computation and iterate over the size of the entire trailing dimension (leading dimension of B) of matrix $A$. To reduce redundant computation, we place a guard condition (see listing \ref{alg:spmm-naive} line 10) to skip iterations where values in $A$ are 0. We can leverage the \ds{} metadata to implement the guard condition in $O(1)$ through the observation that $A[dense_y][dense_x]$ exists iff: $$dense_x \% 
 (\frac{1}{AffineIndices[dense_y].a}) == 0 \wedge dense_x + AffineIndices[dense_y].b >= 0$$

Nevertheless, listing \ref{alg:spmm-naive} has two issues. (1) In SpMM kernels, non-zero values in the product of $AB$ are produced only when non-zero values from $A$ multiply with non-zero values from $B$. However, by letting the loop in line 4 iterate across the entire trailing dimension of $A$ (leading dimension of $B$), we end up with identical loop counts regardless of the degree of sparsity in $A$. (2) The predicate in line 10 will result in control divergence between threads in a warp when corresponding attempts to read values from $A$ are non-zero for certain threads, but zero for others. Moreover, this divergence is exacerbated by being placed within a loop that may run for 1000s of iterations. We propose optimizations to mitigate each issue in section \ref{spmm:optimizations}. 

\label{section:spmm}
\begin{figure*}[t]
    \centering
    \includegraphics[width=\textwidth]{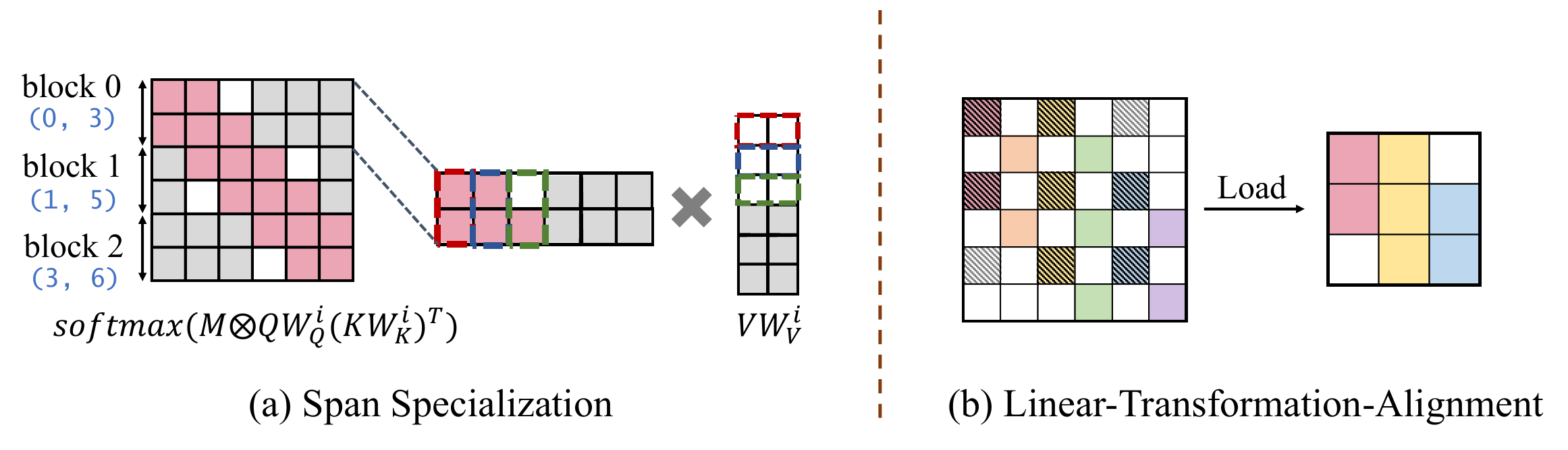}
    % \captionsetup{skip=1pt}
    \caption{Different SpMM optimizations. (a) span-specialisation, white are redundant reads. (b) - linear-transformation-alignment, colors represent points whose affine-indexes are identical (aligned), hashed boxes correspond to values loaded into the same thread-block.}
    \label{fig:spmm}    
\end{figure*}

\subsection{Optimizations}
\label{spmm:optimizations}

\begin{wrapfigure}{R}{0.5\textwidth}
    \begin{minipage}{0.47\textwidth}
        \input{code-listings/spmm-kernel} 
    \end{minipage}
\end{wrapfigure}

\textbf{Span specialisation.} To mitigate issue 1, we observe that sparse-MHSA structures tend to have chunks of non-zero values. For example, the windowed pattern contains all its non-zero values in chunks surrounding the matrix's main diagonal. Therefore, a thread-block reading a selection of rows from a sparse matrix $A$ in the window pattern need not iterate across all dense indices in the trailing dimension and can start at the first non-zero value and end at the last non-zero value: its \textit{column-span}. Suppose a thread-block is reading a collection of rows, $r_1, r_2, ... r_k$ from sparse matrix $A$. Its column-span can be computed in $O(k)$ time through: $$span(r_1, r_2, ... r_k) = [\min_{r_i\in[r_1, r_k]}(-AI[r_i].b), \max_{r_i \in [r_1, r_k]}(AI[r_i].nnzs - AI[r_i].b)]$$ Where $AI$ are the $AffineIndices$. Figure \ref{fig:spmm} (a) demonstrates this \textit{span-specialisation}. The thread-block that reads values from the first two rows of $A$ iterates from index 0 to index 2 in line 4 of algorithm \ref{alg:spmm-naive}. 

\textbf{Linear transformation Alignment.} %The predicate in listing \ref{alg:spmm-naive} line 5 introduces thread-divergence between threads. Consider the strided-pattern in figure \ref{fig:spmm} (b) with a thread-block of size $3\times 1$ reading values from the first three rows. In the first iteration of line 4, the first and third threads will read non-zero values from the first and third rows of $A$. The second thread's predicate in line 5 will evaluate to false and its control flow will diverge. In the second iteration, the second thread will read a non-zero value from $A$, and the first and third threads' predicate will be evaluated to false in line 5. This pattern will continue across every iteration of the loop, resulting in at most half the threads in a warp diverging in control flow. 
To mitigate issue 2, we observe that if rows in the \ds{} have the same affine-indices, data is placed in identical indices across the trailing dimension. We can re-map threads within a warp to operate on rows with identical affine-indices, ensuring that threads execute the body of the main loop in tandem. Figure \ref{fig:spmm} (b) demonstrates this optimization, when loading from $A$, only 2 out of 9 threads diverge in control flow, as opposed to 4 out of 9 without this optimization.   
%and distinct threads may access the values at common strides, reducing thread-divergence. In figure \ref{fig:spmm} (c) we see that we re-map consecutive threads to read values from rows of $A$ with a stride of 2, resulting in two out of 9 threads exhibiting divergent control flow. Typically, this would reduce memory coalescing, but we can re-structure the data-layout of $A$ to reduce un-coalesced accesses in conjunction with this optimization.  

\subsection{\spmm{} code-generation}

\begin{wrapfigure}{R}{0.5\textwidth}
    \begin{minipage}{0.47\textwidth}
        \input{code-listings/spmm-code-gen} 
    \end{minipage} 
\end{wrapfigure}

The \spmm{} code-generation pass, shown in listing \ref{alg:spmm-code-gen}, ingests \ds{} metadata and code-generates a \spmm{} kernel. It returns an \spmm{} kernel (spmmFunc), thread-block count (TBCount), metadata required for optimizations (spmmMetaOpt), and layout of the \ds{} (layout), (see lines 5-7). The optimizations metadata, spmmMetaOpt, contains a map of thread indices to 2 pieces of information. (1) The respective start and end loop indices, implementing span-specialisation. (2) The respective row of the \ds{} to load, implementing linear-transformation-alignment. The boolean layout flag indicates the data-layout the \ds{} should be for correct data indexing. For high-performance, the layout of the \ds{} depends on the density of the input mask, see section ~\ref{splat} for more details.

\section{Final Code-Generation of Sparse-MHSA}
\label{splat}
SPLAT generates high-performance code for end-to-end implementations of sparse-MHSA. Given a pattern, it produces code for the \sddmm{}, softmax, and \spmm{} kernels, implementing an end-to-end sparseMHSA kernel. SPLAT's code-generation mechanism is shown in algorithm \ref{alg:end-to-end-code-gen}. It proceeds in four passes. (1) An analysis pass (see lines 3-5) analyzes the input mask to generate information used by later code-generation passes and ensures the legality of later optimizations and code-generation. (2) A code-generation pass (see lines 7-9), which ingests the information produced by the analysis pass to generate high-performance: \sddmm{}, Softmax, and \spmm{} kernels. (3) A memory allocation and auxiliary data creation pass (see lines 13-15) which allocates enough memory to hold output tensors and creates a data object required for any optimizations for \sddmm{} and \spmm{} kernels. (4) A data-layout reordering pass (see line 15), which reasons about the data-layout of the input tensor to the \spmm{} kernel, inserting the relevant transposition whenever necessary. 

\begin{figure}[t]
    \input{code-listings/splat-code-gen}
\end{figure}

\textbf{Analysis pass} The analysis pass first checks if the mask is regular (see line 3), terminating otherwise. It then generates the affine-indices by solving a set of linear equations (see line 4) as described in section \ref{acsr}. Finally, it analyzes the number of non-zero values in the mask (density analysis - see line 5), and if this number is greater than a threshold $\alpha$, sets the classification returned to dense, else to sparse. This information is required for data-layout re-ordering optimizations for good end-to-end sparseMHSA performance. 

\textbf{Code-generation} The code-generation passes (see lines 7-9) produce high-performance implementations of the \sddmm{} (see section \ref{section:sddmm}) and \spmm{} (see section \ref{spmm}) kernels, as well as objects required to implement optimizations correctly. We use cuDNN's softmax kernel to implement the softmax over the \ds{}. 

\textbf{Memory Allocation \& Auxiliary Data Creation.} Lines 13-15 allocate the necessary amount of memory required to store the output tensors of each of the kernels: \sddmm{}, Softmax, and \spmm{}, and create an auxiliary data object. This data object contains information required for the correctness of the optimizations detailed in sections \ref{section:greedy-tiling} and \ref{section:spmm}. It contains the metadata (the affine-indices and non-zero-values of the \ds{}) required for fast-indexing, the anchor-points for poset-tiling, and a rspmmMetaOpt structure that contains thread-level mappings for linear-transformation-alignment and span-specialization. 

\textbf{Data-layout reordering.} Reasoning about the data-layout of the \ds{} in the \spmm{} kernel is important for a high-performance implementation of sparse-MHSA. At moderate sparsity levels, reading from an \ds{} within the \spmm{} kernel is more expensive than writing to it within the \sddmm{} kernel. The \sddmm{} kernel only writes to each output value once, but the \spmm{} kernel reads each input multiple times (across multiple thread-blocks). We select the best format for the \spmm{} kernel by considering the global computations and their formats and inserting a transpose kernel before it whenever necessary according to the output of density analysis. For input-masks with high-density, we transpose the \ds{} to a column-compressed \& column-major layout, while for input-masks with low-density, we transpose the \ds{} to a row-compressed \& row-major layout before the \spmm{} kernel. The column-compressed \& column-major allows threads to issue coalesced memory requests but requires complex arithmetic to index compared to the row-compressed \& row-major layout. At higher density levels, these un-coalesced requests bottleneck kernels as the amount of data read is greater. Our ablations in section ~\ref{section:ablations} illustrate this.

\section{Evaluation}
\label{Evaluation}
We evaluate SPLAT against state-of-the-art vendor-libraries (SOTA) and hand-optimized implementations across a \textit{variety} of sparse patterns to demonstrate SPLAT's \textit{generality} and \textit{high-performance}. To this end, we conduct a series of run-time performance studies (section \ref{section:comparitive-study}) and ablations (section \ref{section:ablations}). We perform run-time performance studies at different granularities: individual sparse-kernels (\sddmm{} \& \spmm{}), single layer sparse-MHSA, and end-to-end transformer, to demonstrate the efficacy of SPLAT's code-generation framework. Our ablations identify the importance of the optimizations we have proposed.

In summary, our results show that SPLAT exhibits considerable speedups against vendor-libraries and hand-optimized implementations over a variety of sparse-MHSA patterns. SPLAT can achieve speedups of up to 2.07x and 5.68x over \cublas{} and \cusparse{} across desired sparsity ranges of [10\%, 50\%], respectively, further, it achieves up-to 2.05x and 4.05x over hand-written kernels in \triton{} and TVM, respectively.  

\subsection{Implementation}

We implement SPLAT's GPU code-generation mechanism in C++ and Python. Currently, SPLAT is a CUDA code generation system compatible with JAX. SPLAT outputs CUDA implementations of sparse-MHSA that are just-in-time compiled through JAX's CUDA compatible foreign-function-interfaces (FFIs). We use the following software versions: cudatoolkit 11.6, jaxlib 0.4.6, triton 2.1.0, TVM 0.6.0, and pytorch 2.1.0.

\subsection{Experimental Setup}

We evaluate SPLAT's effectiveness across 3 patterns: windowed, blocked and strided. We select these patterns due to their popularity in the deep-learning community ~\cite{gpt-3, longformer, sparse-transformer, reformer} as well as the availability of highly-optimised hand-written implementations of each using various SOTA tensor-algebra compilers (TVM and triton) and deep-learning frameworks (JAX). We compare SPLAT generated code to SOTA implementations of each pattern at three different granularities: individual sparse-kernel primitives, single-layer sparse-MHSA \& an end-to-end transformer, demonstrating speedups across each case.

\textbf{Sparse-kernel Primitives} We compare SPLAT's code-generated \sddmm{} and \spmm{} kernels for each pattern to highly optimised vendor libraries: \cusparse{} and \cublas{}. 

\textbf{Single-layer sparse-MHSA} We compare SPLAT's code-generated single-layer sparse-MHSA kernel for each pattern to SOTA tensor algebra compilers: triton and TVM. Triton is a mature deep-learning compiler with block-sparse kernel APIs that are faster than \cusparse{}, serving as a fair comparison. 

\textbf{End-to-End Transformer} We implement an end-to-end transformer using SPLAT's code-generated sparse-MHSA kernel for each pattern and compare against end-to-end sparse-transformers written in TVM, and JAX. 

Unless otherwise stated, all our comparisons are conducted at the density levels: [0.4, 0.8, 1.6, 3, 6, 12, 24, 44, 75, 100]. We can tune the sparsity of each pattern by varying the width of the stride (strided-pattern), the size of the window (windowed pattern), or the size of the block (blocked pattern). We note that all the models we compare to are hand-written solutions, specialized to a single sparse-MHSA pattern. SPLAT, on the other hand, is general enough to generate high-performance code for all the patterns.

\vspace{-5pt}
\subsection{Run-time Performance Study}
\label{section:comparitive-study}
\begin{figure}[t]
    \centering
    \includegraphics[width=\linewidth]{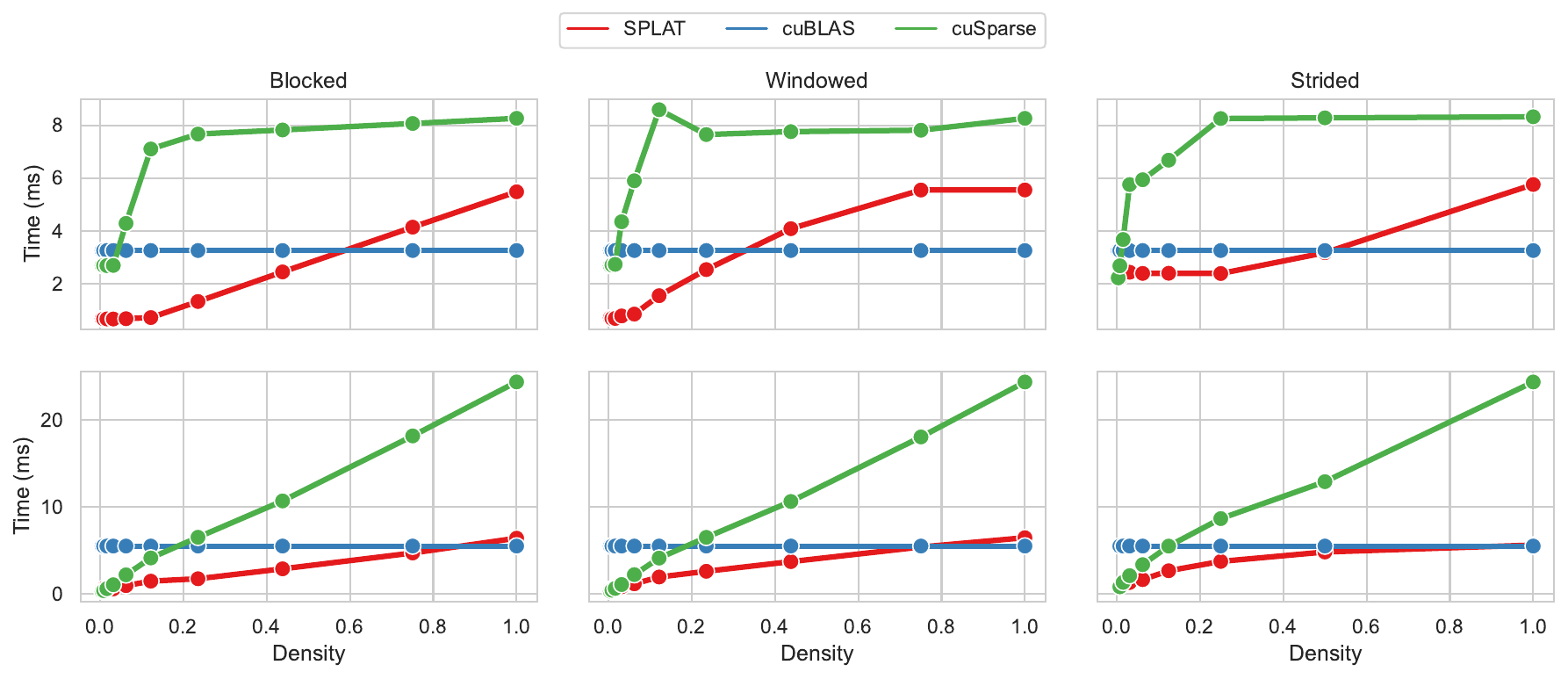}
    \caption{Run-time performance of sparse-primitives: \sddmm{} and \spmm{}, comparing: SPLAT, \cublas{} and \cusparse{}. The top and bottom rows are the \sddmm{} and \spmm{} results, respectively. The desired density levels observed in sparse-MHSA are in the [10,50]\%.}
    \label{fig:comparitive-evaluation-micro}    
    % \vspace{-15pt}
\end{figure}

The primary motivation of the run-time performance study is to answer the question: Can SPLAT be used to accelerate end-to-end sparse-MHSA-based models, and is this speedup \textit{as a result} of SPLAT's code-generation methodology? To answer the first part of the question, we evaluate SPLAT at three levels of granularity. To answer the second part of the question, we analyze the memory profiles of individual kernels and conduct a breakdown analysis of single layer sparse-MHSA showing this speedup is a result of SPLAT's code-generation mechanism.

%First, we evaluate SPLAT's individual sparse-primitives against highly optimized vendor-libraries, second we evaluate SPLAT's end-to-end sparse-MHSA generated code against hand-written implementations in \triton{} and TVM, and third we evaluate an end-to-end sparse transformer, replacing its sparse-MHSA layer with SPLAT's generated code, against hand-written sparse transformer implementations.

\subsubsection{Individual Kernels}

We evaluate SPLAT's sparse primitive implementations of \sddmm{} and \spmm{} against \cusparse{} and \cublas{} to ascertain whether SPLAT's sparse-MHSA primitives outperform vendor-libraries at the desired sparsity ranges of (10\%-50\%). We compare across three patterns: blocked, windowed, and strided.

\textbf{Speedups.} Figure \ref{fig:comparitive-evaluation-micro} shows our results for runtime performance. For the \sddmm{}, SPLAT experiences geomean speedups of 2.46x \& 5.68x (blocked pattern), 1.29x \& 3.17x (windowed pattern), 1.24x \& 2.93x (strided pattern) over \cublas{} and \cusparse{} respectively. 
%\charith{across the density ranges}. 
For the \spmm{}, SPLAT experiences geomean speedups of 2.81x \& 3.37x (blocked pattern), 2.07x \& 2.47x (windowed pattern), 1.51x \& 2.33x (strided pattern) over \cublas{} and \cusparse{} respectively. All these speedups are reported in the 10\%-50\% density range. %\charith{shall we say the cut-off points as well?}

\subsubsection{Single layer Sparse-MHSA}
\begin{figure}[t]
    \centering
    \includegraphics[width=\linewidth]{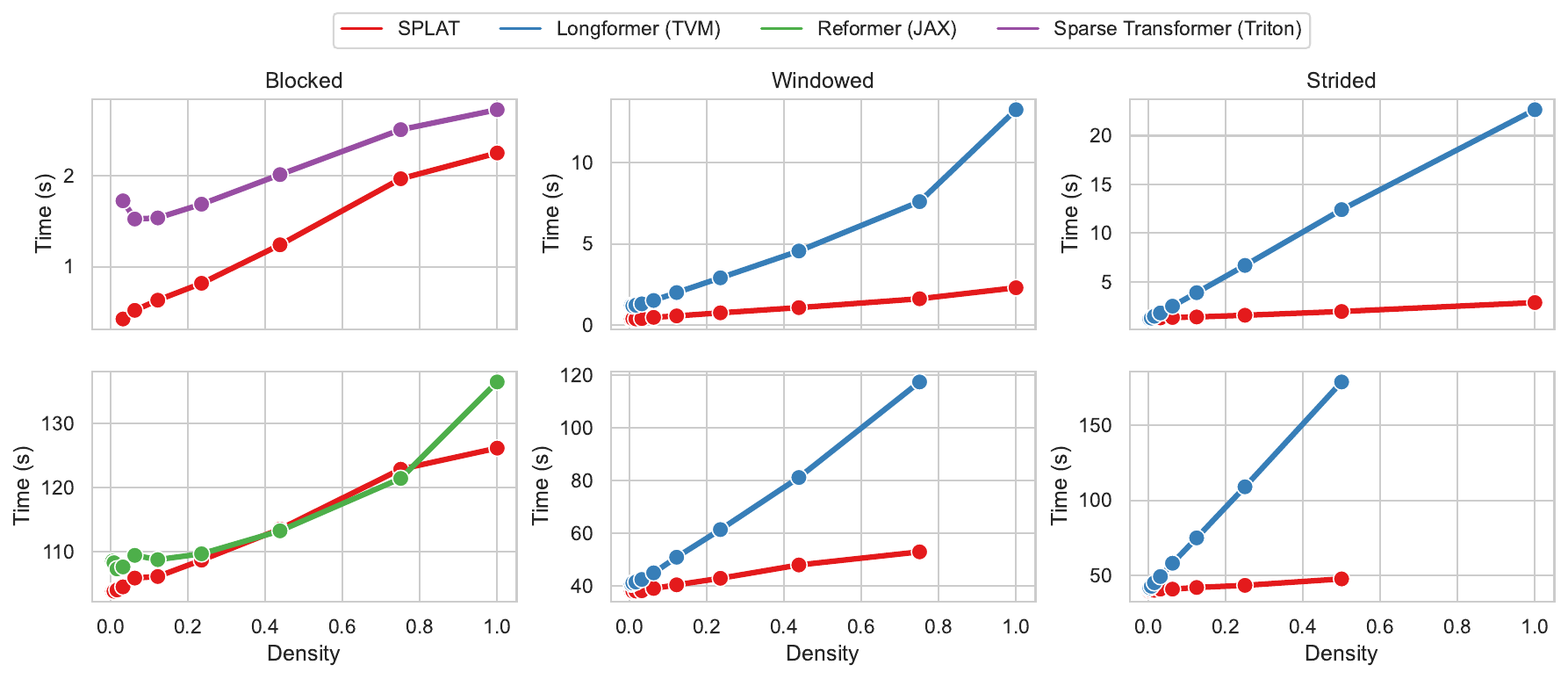}
    \caption{Run-time performance of a single-layer MHSA and end-to-end sparse transformer, comparing SPLAT against sparse transformer (implemented in \triton{}), longformer (implemented in TVM) and reformer (implemented in JAX). The top and bottom rows are the single-layer sparse-MHSA and end-to-end transformer implementations, respectively. Unplotted points at densities of 1.0 in middle and bottom (right) are due to OOM issues.}
    \label{fig:comparitive-evaluation-macro}    
    \vspace{-15pt}
\end{figure}

We evaluate SPLAT's end-to-end sparse-MHSA generated code against hand-written implementations of a single layer of sparse-MHSA written in \triton{} and TVM to ascertain whether SPLAT effectively implements an end-to-end sparse-MHSA solution that reasons about data-layout conversion costs \textit{globally}. 
We compare across three patterns: blocked (implemented in \triton{}), windowed, and strided (implemented in TVM).

\textbf{Speedups} Figure \ref{fig:comparitive-evaluation-macro} (top row) shows our results for runtime performance. SPLAT realizes geomean speedups of 2.05x, 4.05x, and 2.12x over triton, TVM-windowed, TVM-strided respectively across the \textit{entire} density range.

\subsubsection{End-to-End Sparse Transformer}

Third, we evaluate an end-to-end sparse transformer, implementing the sparse-MHSA layer with SPLAT's generated code, against hand-written implementations in TVM and JAX to ascertain whether SPLAT can be used to accelerate a full transformer. We compare across three patterns: blocked, windowed, and strided. 

\textbf{Speedups} Figure \ref{fig:comparitive-evaluation-macro} (bottom row) shows our results for runtime performance. SPLAT experiences geomean speedups of 1.03x, 1.31x, and 1.49x over Reformer (blocked pattern) implemented in JAX, and Longformer (windowed and strided pattern) implemented in TVM, respectively across the \textit{entire} density range. 

\subsubsection{Analysis}
\begin{table*}[t]
\centering
\resizebox{\linewidth}{!}{%
\begin{tabular}{l|l|c|cc|cc|cc}
  \toprule
  \multirow{2}{*}{Kernel} & \multirow{2}{*}{Method} & \multirow{2}{*}{Threads/SM} & \multicolumn{2}{c}{Read (GB)} & \multicolumn{2}{|c}{Write (GB)} & \multicolumn{2}{|c}{Cache Hit Rate (\%)} \\
  \cmidrule(lr){4-5} \cmidrule(lr){6-7} \cmidrule(lr){8-9}
  ~ & ~ & ~ & \multicolumn{1}{c}{$\texttt{Global} \rightarrow \texttt{L2}$} & \multicolumn{1}{c}{$\texttt{L2} \rightarrow \texttt{L1}$} & \multicolumn{1}{|c}{$\texttt{L1} \rightarrow \texttt{L2}$} & \multicolumn{1}{c}{$\texttt{L2} \rightarrow \texttt{Global}$} & \multicolumn{1}{|c}{L1} & \multicolumn{1}{c}{L2} \\
  \midrule
  \multirow{5}{*}{SDDMM} & SPLAT & 1152 & 0.190 & 3.170 & 0.377 & 0.346 & 44.25 & 92.91 \\
  ~ & cuBLAS & 512 & 0.201 & 1.610 &  1.610 & 1.590 & 00.41 & 92.34 \\
  ~ & cuSPARSE & 576 & 0.206 & 7.830 & 0.662 & 0.365 & 42.02 & 96.98 \\
  ~ & \triton{} & 128 & 0.201 & 0.432 & 0.681 & 0.360 & 59.21 & 82.25\\
  ~ & TVM (x32) & 2048 & 0.006 & 0.358 & 0.035 & 0.001 & 75.23 & 98.19 \\
  \midrule
  \multirow{5}{*}{SpMM} & SPLAT & 1152 & 0.101 & 0.602 & 0.805 & 0.084 & 57.53 & 91.10 \\
  ~ & cuBLAS & 512 & 1.720 & 2.420 & 0.100 & 0.970 & 00.10 & 43.43 \\
  ~ & cuSPARSE & 1536 & 0.203 & 11.330 & 1.100 & 0.089 & 53.71 & 97.57 \\
  ~ & \triton{} & 128 & 0.482 & 3.880 & 2.090 & 0.126 & 22.24 & 90.33 \\
  ~ & TVM (x32) & 2048 & 0.002 & 0.172 & 0.060 & 102 KB & 89.94 & 92.35 \\
  \bottomrule
\end{tabular}}
 \caption{Memory profiles of SPLAT, \cublas{}, \cusparse{}, \triton{} and TVM of the blocked pattern (except TVM which is the window pattern) at a density level of 24\%. $\texttt{Global} \rightarrow \texttt{L2}$ and $\texttt{L2} \rightarrow \texttt{L1}$ is the amount of data transferred from global-memory to L2 cache, and L2 to L1 cache, respectively, as a result of memory reads. $\texttt{L1} \rightarrow \texttt{L2}$ and $\texttt{L2} \rightarrow \texttt{Global}$ is the amount of data transferred from L1 to L2 cache, and L2 to global-memory, respectively, as a result of memory writes. TVM spawns 32 kernels, hence the (x32) notation.}
 \label{table:profile-information-micro}
 \vspace{-10pt}
\end{table*}

% \vspace{-10pt}
We now analyze how SPLAT's sparse-primitives achieve speedups over optimized vendor-libraries and hand-written kernels in \triton{} and TVM. We show that SPLAT's novel code-generation algorithms leverage the meta-data stored in the ACSR effectively to produce favorable memory access and write patterns balanced with enough inter-warp parallelism to hide read/write latencies. We show this by analyzing the memory profiles of all vendor-libraries, hand-written kernels, and SPLAT at a density level of 24\% for the blocked pattern. Favorable memory access patterns will read similar amounts of data from global memory to L2 cache, and from L2 to L1 cache, reducing extraneous data-movement through the memory hierarchy; we compute how much more data is transferred from L2 to L1 (denoted as $\texttt{L2} \rightarrow \texttt{L1}$), compared to global memory to L2 (denoted as $\texttt{Global} \rightarrow \texttt{L2}$) for all hand-written kernels, libraries, and SPLAT. We apply a similar argument to memory write-back patterns, computing the excess data written from L1 to L2 ($\texttt{L1} \rightarrow \texttt{L2}$), compared with L2 to global memory ($\texttt{L2} \rightarrow \texttt{Global}$). Our results are in table \ref{table:profile-information-micro}. We systematically compare SPLAT to each state-of-the-art vendor library and hand-written kernel. 

\begin{wrapfigure}[18]{R}{0.5\textwidth}
%\begin{figure}[H]
    \includegraphics[width=0.48\textwidth]{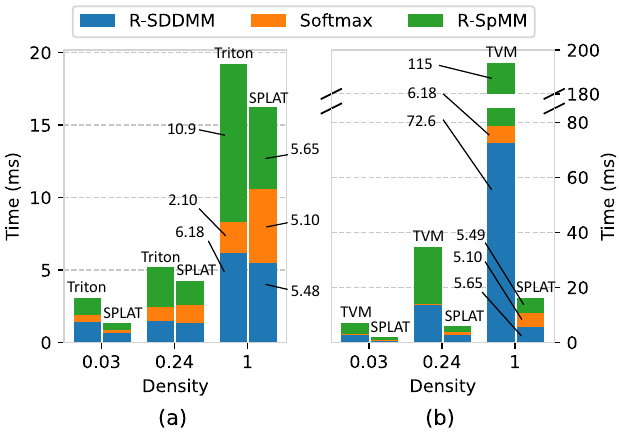}
    \caption{A breakdown analysis of the three components: \sddmm{}, Softmax and \spmm{} of SPLAT, \triton{} and TVM's sparse-MHSA primitives. (a) is the blocked pattern, (b) is the windowed pattern.}
    \label{fig:breakdown}
%\end{figure}
\end{wrapfigure}

\textbf{Vendor-libraries} Analyzing thread access patterns, we report the excess data moved across the memory hierarchy due to the reading of data: 2.98GB (0.501GB), 1.409GB (0.7GB), and 7.624 (11.127GB) for SPLAT, \cublas{}, and \cusparse{} respectively for the \sddmm{} (\spmm{}) kernel. We observe SPLAT's thread-access patterns move significantly less data across the memory hierarchy compared to \cusparse{}, and slightly more compared to \cublas{}. We note \cublas{}, as a dense mat-mul, has regular thread access patterns indexing dense 2-D arrays (as opposed to complex sparse structures), and is thus amenable to favorable access patterns. Nevertheless, SPLAT spawns more threads per streaming-multiprocessor (SM), thus effectively latency hiding expensive memory read operations through inter-warp parallelism. Since \cublas{}'s kernel is compute-bound, there is enough reuse to circumvent the need to latency hide memory reading costs. 

We similarly report the excess data moved across the memory hierarchy as a result of write-backs: 0.031GB (0.721GB), 0.02GB (0.87GB), and 0.257GB (1.01GB) for SPLAT, \cublas{}, and \cusparse{} respectively for the \sddmm{} (\spmm{}) kernel. We observe that the write-back pattern profile of SPLAT is comparable to \cublas{}, and moves significantly less extraneous data compared to \cusparse{}. Overall, since SPLAT computes less than 1/4th of the values compared to \cublas{}, and has favorable access/write-back patterns it is the fastest of the three. 

\textbf{Hand-written Kernels} Analyzing thread access patterns, we report the excess data moved across the memory hierarchy: 2.98GB (0.501GB), 0.231 (3.398GB), and 0.352GB (0.17GB) for SPLAT, \triton{} and TVM respectively for the \sddmm{} (\spmm{}) kernel. We observe that SPLAT's access patterns are better than \triton{}'s \spmm{} and both of TVM's kernels (TVM spawns 32 kernels, one for each batch, thus operates on 1/32nd the amount of data compared to SPLAT and \triton{}). Though \triton{}'s \sddmm{} access patterns are slightly better than SPLAT's, it spawns 9x fewer threads per SM, inadequately hiding read/write latencies. A closer inspection of the kernel indicates this is a result of overusing shared-memory. 

\begin{wrapfigure}[14]{R}{0.5\textwidth}
    \includegraphics[width=0.48\textwidth]{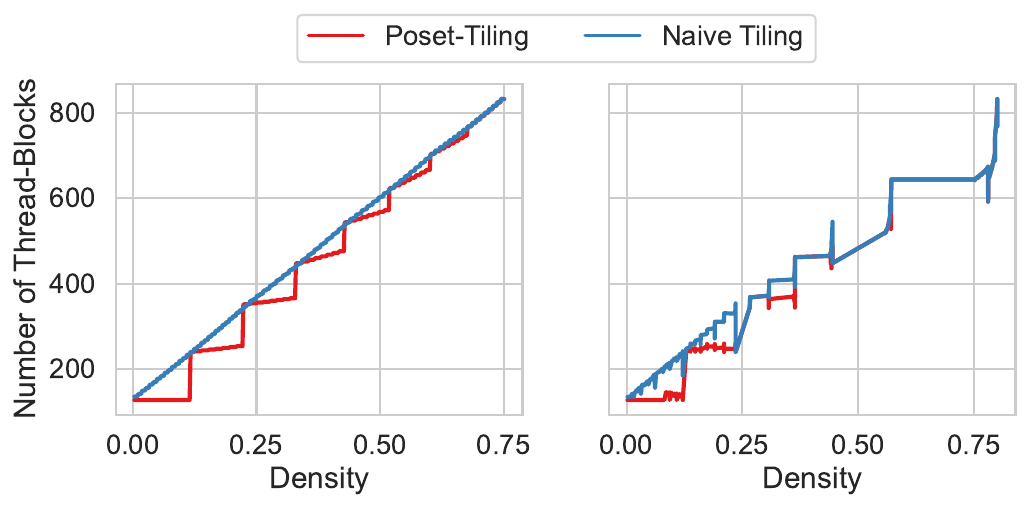}
    \caption{A comparison of the number of thread-blocks used between different Tiling Strategies for the window (left) and blocked (right) pattern. Lower is better.}
    \label{fig:sddmm-tiling-ablation}
\end{wrapfigure}

Similarly, we report the excess data moved through the memory hierarchy as a result of write-backs: 0.031GB (0.721GB), 0.321GB (1.964GB), and 0.034GB (0.06GB) for SPLAT, \triton{} and TVM, respectively for the \sddmm{} (\spmm{}) kernel. We observe that SPLAT's write-back patterns are better than both \triton{} and TVM's.  

\textbf{Breakdown Analysis} To show that end-to-end sparse-transformers are accelerated due to SPLAT's high-performance code-generation mechanism, we break down the run-times of SPLAT's \sddmm{}, softmax, and \spmm{} kernels in a single sparse-MHSA layer and compare it to triton and TVM in \ref{fig:breakdown}. We breakdown these run-times across high, moderate, and low sparsity levels. We see that across all sparsity levels, the collective run-time of SPLAT's kernels is faster than \triton{} and TVM's. 

%For SPLAT to generate high-performance code for an end-to-end sparse-MHSA layer, it needs to reason about data-layout conversion costs \textit{globally}, generating an implementation of a collection of kernels with overall low run-time. To demonstrate this, we breakdown the run-times of SPLAT's \sddmm{}, softmax, \spmm{} kernels in a single layer sparse-MHSA and compare it to \triton{} and TVM in figure \ref{fig:breakdown}. We breakdown these run-times across high, moderate, and low levels of sparsity. We see that across all sparsity levels, the collective run-time of the SPLAT's kernels is faster than \triton{} and TVM's. Of particular note is \triton{}'s softmax, which is faster than SPLAT's softmax as a result of SPLAT algorithmically selecting a locally non-optimal data-layout (row-compressed column-major) for a globally efficient \textit{collection} of kernels. To coalesce the write-back patterns of the \sddmm{} kernel, the \ds{} should be row-compressed and column-major. Un-coalescing writes for a more optimal softmax induces a 3x slow-down in the \sddmm{} (see table \ref{table:l2-mem-write-back} \dvj{Wrong ref?}), a more expensive kernel.

\subsection{Ablation Studies}
\label{section:ablations}

\subsubsection{\sddmm{} Tiling} %\ahan{State what naive tiling is and how you compare against that. These ablations should have runtimes as well (and not just our own metrics.). Then we selected 10 density levels to showcase that a reduction in thread-blocks actually manifests in a reduction in time (Run 10\% maybe?).} 
High-performance arrangements use a minimal number of tiles. We compare the number of thread-blocks used in poset-tiling against a Naive tiling approach for the blocked and windowed pattern. We use a sequence length of 1024 and vary the density of each pattern across all possible values. The results are in figure \ref{fig:sddmm-tiling-ablation}. Poset tiling reduces the number of thread-blocks by 1.098 and 1.095 for the window and blocked pattern respectively, on average. The maximum reduction is 1.83 and 1.72 which uses 13568 and 12544 fewer threads for the window and blocked pattern respectively. %As the patterns progressively become more dense, the increase in the number of thread-blocks used follows a step-wise function for poset-tiling. This occurs because arrangements of tiles usually have divergent threads which do no useful work due to boundary conditions. As the density increases, the number of divergent threads in an arrangement decreases until eventually new thread-blocks are required for points that exceed their boundaries. 

\begin{wrapfigure}[14]{R}{0.5\textwidth}
    \includegraphics[width=0.48\textwidth]{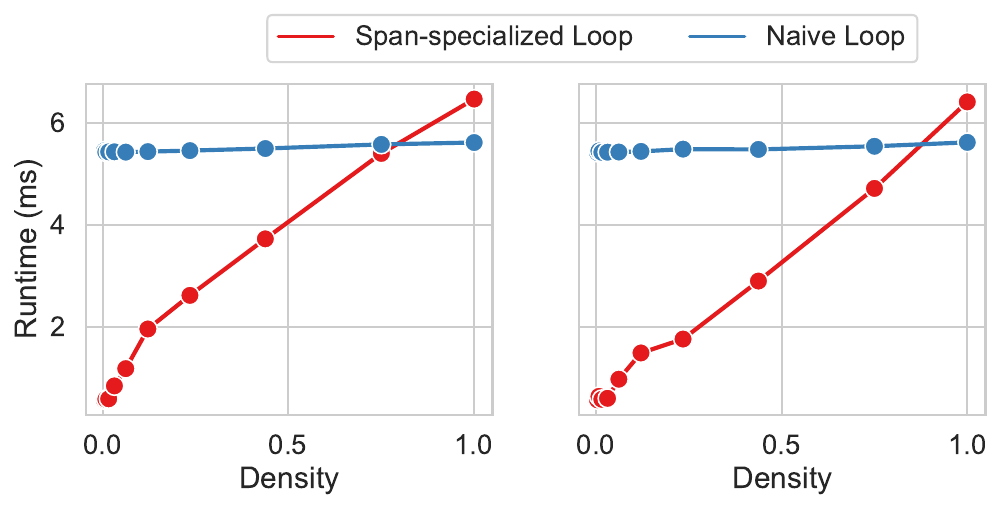}
    \caption{A comparison between the runtimes of the \spmm{} kernel with and without span-specialisation. Left and right are the window and blocked patterns respectively.}
    \label{fig:span-specialisation-ablation}
\end{wrapfigure}

\subsubsection{\spmm{} Optimisations}

We evaluate the benefit of the optimizations: span-specialisation and linear-transformation alignment on \spmm{} kernels at varying density levels by comparing these optimizations to implementations where they are disabled.

Figure \ref{fig:span-specialisation-ablation} shows the results for the effects of span-specialisation on the runtime of \spmm{} kernels. We fix a sequence length of 1024 and vary the density of the window and blocked pattern in [0.4, 0.8, 1.6, 3, 6, 12, 24, 44, 75, 100]. Across these density levels, span-specialization results in geomean speedups of: 3.4x and 3.96x for the windowed and blocked pattern respectively. This shows for the density ranges observed in sparse-MHSA [10,50]\%, span-specialization achieves speedups. However, for extremely dense inputs where loop counts span the entire trailing dimension, span-specialization can be costly due to extra integer arithmetic to compute the loop start and end indices.
%However, when inputs are extremely dense and loop counts are similar, span-specialization requires extra integer arithmetic to compute the loop start and end indices resulting in more work for relatively smaller gains.
%results in a marginal decrease in performance. This occurs because dense inputs have similar loop counts, however span-specialisation requires extra integer arithmetic to compute the loop start and end indices, resulting in more work for relatively smaller gains. We can predicate this optimization to occur only when inputs are moderately or extremely sparse. 

Figure \ref{fig:transformation-alignment-ablation} shows the results for the effects of linear-transformation alignment on loads from a regularly sparse matrix in the \spmm{} kernel (algorithm \ref{alg:spmm-naive} line 8). We fix a sequence length of 1024 and vary the density for the strided pattern across all possible values (varying the stride from 1 to 1024). We compare to a \spmm{} kernel without this optimization (naive loads). Linear-transformation alignment reduces control divergence of loads by 2.73, on average, and the maximum reduction is 8.1.

\begin{wrapfigure}[19]{R}{0.5\textwidth}
    \includegraphics[width=0.48\textwidth]{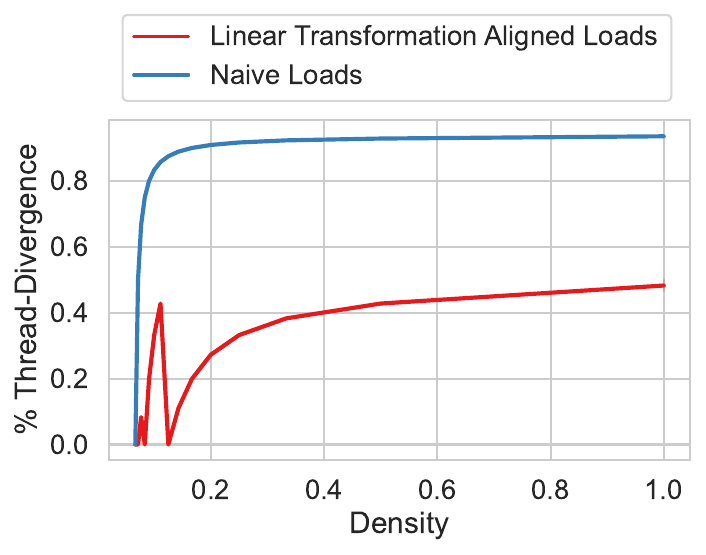}
    \caption{A comparison between the percentage of threads that exhibit control divergence of loads to a regularly sparse matrix in the \spmm{} kernel for the strided pattern (figure \ref{fig:regular-sparsity} middle). Lower is better.}
    \label{fig:transformation-alignment-ablation}
\end{wrapfigure}

\subsubsection{Data-Layout Exploration} We evaluate the benefits of different \ds{} layouts in \spmm{} kernels in figure \ref{fig:data-layout-reordering}. We compare two layouts combined with the cost of transposing data into these layouts: row-compressed \& row-major against column-compressed \& column-major as these produce the fastest \spmm{} kernels across various density levels. We compare the runtimes of the \spmm{} kernels with these layouts across density levels [0.4, 0.8, 1.6, 3, 6, 12, 24, 44, 75, 100]. We categorize these density levels into sparse inputs (density <10\%) and dense inputs (density $\geq$10\%).

\textbf{Sparse Inputs} Across density levels lower than 10\%, the \spmm{} kernel which uses the row-compressed \& row-major layout experiences a geomean speedup of 1.37x compared to a column-compressed \& column-major layout. The row-compressed \& row-major layout requires less complex index arithmetic to reference data at the cost of un-coalesced accesses to non-zero values which happens due to contiguous elements in a column being placed far apart in memory. At these sparsity levels, the memory bandwidth has not yet saturated, hence the cost of un-coalesced accesses is relatively low, resulting in faster \spmm{} kernels in this layout.  

\textbf{Dense Inputs} Across density levels greater than 10\% the \spmm{} kernel which uses the column-compressed \& column-major layout experiences a geomean speedup of 1.6x compared to a row-compressed \& row-major layout. The column-compressed \& column-major layout enables coalesced references of data, since contiguous elements in a column are placed in contiguous memory addresses, at the cost of more complex arithmetic to reference this data. For dense inputs where we read a lot of data, the performance benefits of memory coalescing outweigh that of simplified arithmetic, resulting in faster \spmm{} kernels in this layout. 

\begin{wrapfigure}{R}{0.5\textwidth}
    \includegraphics[width=0.5\textwidth]{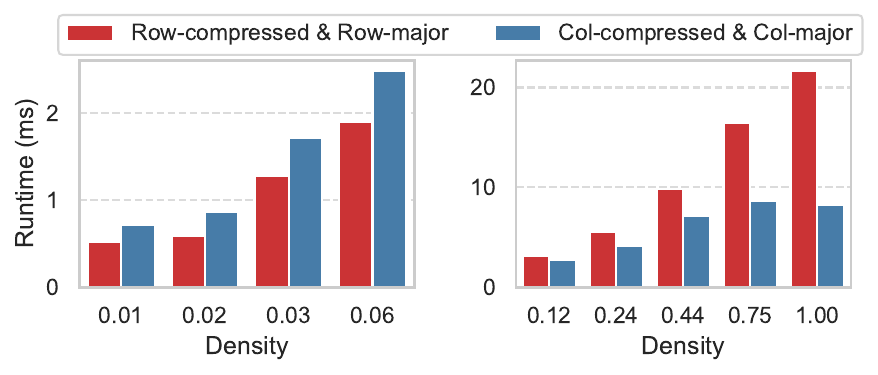}
    \caption{A comparison of the combined runtimes of the transposition and \spmm{} kernels for different \ds{} layouts at various density levels. We use the window pattern. Lower is better.}
    \label{fig:data-layout-reordering}
\end{wrapfigure}

\section{Related Work}
\label{Related Work}
\textbf{Polyhedral Compilation}  Polyhedral compilers \cite{tiramisu, polly, polymage, tensor-comprehensions, pluto} are capable of source-to-source translation of sequential C code into parallel-code (e.g. OpenMP). They work by applying the polyhedral model, decomposing loops into iteration spaces, and analyzing their data-dependencies to parallelize computation. Nevertheless, polyhedral analysis is limited to codes where index arithmetic are affine functions of loop variables. Because sparse tensor algebra code indexes sparse structures through various indirections, polyhedral analysis is ineffective in uncovering parallelization and tiling strategies.

\textbf{Sparse Tensor Algebra Compilers} Popular sparse tensor algebra compilers like TACO \cite{taco}, SparseTIR \cite{sparsetir} are capable of emitting parallel CPU and GPU codes for sparse tensor arithmetic. Nonetheless, the optimisations suggested in SPLAT are not available in these tensor-compilers and hence cannot generate high-performance code for a variety of sparse-MHSA patterns. 

%Nonetheless, they mainly target high levels of sparsity and only operate on a finite set of sparse formats (e.g. CSR). In doing so, sparse tensor algebra compilers cannot express data in our ACSR format, and cannot take advantage of the optimizations the ACSR make amenable.

\textbf{Structured Sparsity} NVIDIA introduced a new specialized data-path in the Ampere \cite{ampere} architecture to compute \textit{structured-sparsity} \cite{structured-sparsity}. These instructions and special hardware units compute 1:2, 2:4 structured-sparsity where the largest absolute value of every 2 (or two largest values out of every 4) elements is kept, and the rest are pruned. Although these instructions target moderate sparsity levels, they cannot represent any of the sparse-MHSA patterns since structured-sparsity is input dependent, resulting in a random distribution of non-zero values. 

\textbf{Sparse Formats} There have been a wide variety of sparse-formats proposed in the literature, please look at the survey in \cite{spmm-survey} for further details (e.g. CSR, COO, ELLPACK, DCSR, DIA, BCSR, CSB, CSF to name a few). Nevertheless, traditional sparse formats are too general, and custom sparse formats (e.g. BCSR, DIA) are too specific to implement high-performance kernels for a variety of sparse-MHSA patterns.

\section{Conclusion}
\label{Conclusion}
We have described SPLAT, an optimized code-generation framework that targets a variety of sparse-MHSA patterns. SPLAT introduces novel tiling strategies and optimizations for better thread access patterns, and reasons about data-layout conversion costs in a \textit{global} manner. SPLAT exploits the regular nature of sparse-MHSA patterns, introducing a new sparse-format: \ds{}, that enables SPLAT's code-generation schemes to have favorable memory-access patterns. We use SPLAT to implement a variety of sparse-MHSA patterns and transformers, demonstrating its generality and high-performance. Our experiments show that SPLAT realizes geomean speedups of 2.05x and 4.05x over hand written kernels written in \triton{} and TVM respectively.

\section*{Acknowledgements}
We would like to thank Jai Arora and Wanyu Zhao for their feedback on early drafts of this paper. This work is partly supported by ACE, one of the seven centers in JUMP 2.0, a Semiconductor Research Corporation (SRC) program sponsored by DARPA.

%%
%% The next two lines define the bibliography style to be used, and
%% the bibliography file.
\bibliographystyle{ACM-Reference-Format}
\bibliography{sample-base}

%%
%% If your work has an appendix, this is the place to put it.
\newpage
\appendix
\section*{Appendix}
\appendix
\setcounter{definition}{6}
\setcounter{theorem}{2}

\section{Stretching in Polygonal Patterns}
\label{supp:stretch-poly}

\begin{theorem}
    If we assume that the anchor of the thread block $TB_i$ is present in a polygonal mask, i.e. $Anc(TB_i) \in P$, then $| Cov(TB_i) |$ would be maximum when its stretch factor $Str(TB_i) = 1$.
\end{theorem}

\begin{proof}
    The argument is pretty simple here.
    
    If $| Cov(TB_i) | = mn$ for $| Str(TB_i) | = 1$, then it is already at its maximum since $| Cov(TB_i) | \le | Comp(TB_i) | = mn$.

    If $| Cov(TB_i) | < mn$ for $| Str(TB_i) | = 1$, then one of the right edge or the bottom edge is crossing the polygonal mask boundary. WLOG, let us assume that the right edge is outside the polygonal mask. If we were to increase the stretch of the thread block, we would only find more points not belonging to the mask. Thus, the $| Cov(TB_i) |$ will never increase beyond $| Str(TB_i) | = 1$.
\end{proof}

\section{Stretching in Strided Patterns}
\label{supp:stretch-strided}

\begin{theorem}
    If we assume that the anchor of the thread block $TB_i$ is present in a strided mask, i.e. $Anc(TB_i) \in P$ and that the stride $X$ is divisible by its stretch factor $s = Str(TB_i)$, i.e. $X = \kappa s$, we can precisely calculate: $$| Cov(TB_i) | = \bigg\lceil \frac{m - n + 1}{\kappa} \bigg\rceil \times n + \sum_{k = \lceil \frac{m - n + 1}{\kappa} \rceil}^{\lceil \frac{m - 1}{\kappa} \rceil} (m - k \kappa) + \sum_{k = 1}^{\lceil \frac{n - 1}{\kappa} \rceil} (n - k \kappa) $$
\end{theorem}

\begin{proof}
    Let's assume that point $p$ is present in both $Comp(TB_i)$ and $P$. Let $(x,y) = Anc(TB_i)$.\\
    Since $(x, y) \in P$, $(x, y) = (\alpha_0 \kappa s + t_0, \beta_0 \kappa s + t_0)$ for some $\alpha_0, \beta_0 \in \mathbb{N}_0$, $t_0 \in \mathbb{Z}_{\kappa s}$ \\
    Since $p \in Comp(TB_i)$, $p = (x + is, y + js)$ for some $i \in \mathbb{Z}_n, j \in \mathbb{Z}_m$ \\
    Since $p \in P$, $p = (\alpha_1 \kappa s + t_1, \beta_1 \kappa s + t_1)$ for some $\alpha_1, \beta_1 \in \mathbb{N}_0$, $t_1 \in \mathbb{Z}_{\kappa s}$ \\
    Therefore, $x + is = \alpha_1 \kappa s + t_1$ and $y + js = \beta_1 \kappa s + t_1$. \\
    $\implies x + is - \alpha_1 \kappa s = y + js - \beta_1 \kappa s$ \\
    $\implies \alpha_0 \kappa s + t_0 + is - \alpha_1 \kappa s = \beta_0 \kappa s + t_0 + js - \beta_1 \kappa s$ \\
    $\implies \alpha_0 \kappa s + is - \alpha_1 \kappa s = \beta_0 \kappa s + js - \beta_1 \kappa s$ \\
    $\implies \alpha_0 \kappa + i - \alpha_1 \kappa = \beta_0 \kappa + j - \beta_1 \kappa$ [Since $s \neq 0$] \\
    $\implies \kappa (\alpha_0 - \alpha_1 - \beta_0 + \beta_1) = j - i$ \\
    $\implies \kappa$ divides $(j - i)$ [Since $\alpha_0, \alpha_1, \beta_0, \beta_1 \in \mathbb{N}_0$]
    \medskip
    
    We know that $| Comp(TB_i) | = mn$ \\
    $| Cov(TB_i) |$ is the number of points present in both $Comp(TB_i)$ and $P$ which is equivalent to number of solutions of $(i, j) \in \mathbb{Z}_n \times \mathbb{Z}_m$ satisfying $\kappa$ divides $(j - i)$.
    \medskip
    
    As shown in Fig~\ref{fig:stride-effectiveness-proof-plot}, we can partition the thread block into three sections. \\
    Green section contains threads with $(j - i) \in \{ -n, \dots, -1 \}$. \\
    Pink section contains threads with $(j - i) \in \{ 0, \dots, m-n \}$. \\
    Blue section contains threads with $(j - i) \in \{ m-n+1, \dots, m \}$.

    $$\text{Number of solutions in green section} = \sum_{k = 1}^{\lceil \frac{n - 1}{\kappa} \rceil} (n - k \kappa) $$
    $$\text{Number of solutions in pink section} = \bigg\lceil \frac{m - n + 1}{\kappa} \bigg\rceil \times n  $$
    $$\text{Number of solutions in blue section} = \sum_{k = \lceil \frac{m - n + 1}{\kappa} \rceil}^{\lceil \frac{m - 1}{\kappa} \rceil} (m - k \kappa) $$
    
    Therefore, $$| Cov(TB_i) | = \bigg\lceil \frac{m - n + 1}{\kappa} \bigg\rceil \times n + \sum_{k = \lceil \frac{m - n + 1}{\kappa} \rceil}^{\lceil \frac{m - 1}{\kappa} \rceil} (m - k \kappa) + \sum_{k = 1}^{\lceil \frac{n - 1}{\kappa} \rceil} (n - k \kappa)$$
\end{proof}

\begin{figure}[H]
    \centering
    \includegraphics[scale=0.3]{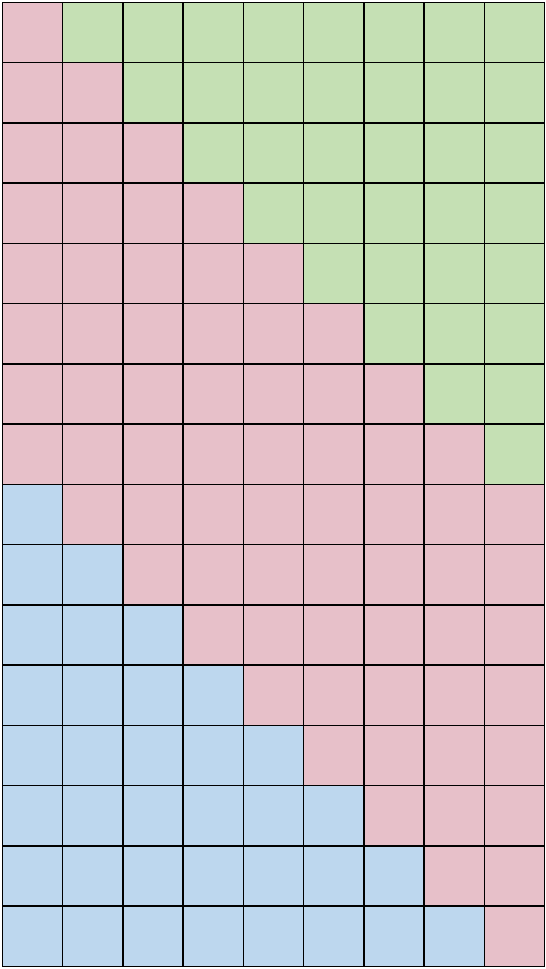}
    \caption{Partitioning the Thread Block based on $j-i$ value}
    \label{fig:stride-effectiveness-proof-plot}
\end{figure}

\begin{theorem}
    Stronger form: If we assume that the anchor of the thread block $TB_i$ is present in a strided mask, i.e. $Anc(TB_i) \in P$ with stretch factor $s = Str(TB_i)$, we can precisely calculate: $| Cov(TB_i) | = f(\kappa)$ where $$f(\kappa) = \bigg\lceil \frac{m - n + 1}{\kappa} \bigg\rceil \times n + \sum_{k = \lceil \frac{m - n + 1}{\kappa} \rceil}^{\lceil \frac{m - 1}{\kappa} \rceil} (m - k \kappa) + \sum_{k = 1}^{\lceil \frac{n - 1}{\kappa} \rceil} (n - k \kappa) \qquad \text{and} \qquad \kappa = \frac{X}{gcd(s, X)}$$
\end{theorem}

\begin{proof}
    Similar to previous proof, we can conclude $(\alpha_0 - \alpha_1 - \beta_0 + \beta_1) d = (j - i) s$. \\
    $\implies \frac{(j-i)s}{d} \in \mathbb{Z}$. \\
    Let $\tau = gcd(s, X)$, i.e. $X = \kappa \tau$, $s = \eta \tau$, $gcd(\kappa, \eta) = 1$. \\
    $\implies \frac{(j-i)\eta}{\kappa} \in \mathbb{Z}$. \\
    $\implies \kappa$ divides $(j-i)$

    The rest of the proof remains the same.
\end{proof}

\begin{theorem}
    Suppose we have a strided pattern with stride $X$. $\lambda^s$ is the number of thread-blocks generated by poset tiling with stretch factor $s$. Then we have that $\lambda^s$ decreases as $\gcd(s, X)$ increases.
\end{theorem}

\begin{proof}
    By definition of strided pattern and poset tiling algorithm, we have no overlapping thread blocks. Thus, we can effectively compute $\lambda^s = \frac{|P|}{f(\kappa)}$. All terms in $f(\kappa)$ are monotonically decreasing with $\kappa$. $\lambda^s$ monotonically increases with $\kappa = \frac{X}{gcd(s, X)}$. Thus, $\lambda^s$ decreases as $gcd(s, X)$ increases.
\end{proof}

For strided pattern, $Cost(TB_{poset}) = \lambda^s s$. Since, $\lambda^s$ decreases as $gcd(s, X)$ increases. $Cost(TB_{poset})$ will achieve maximum value for some $s \in factors(X)$, thereby, reducing the search space of stretch factor.

\section{Proof of Theorem 1}
\label{supp:theorem1}

In Theorem 1, we make claims on the optimality bounds of the Poset Tiling Algorithm. In \ref{supp:bound-poly}, we prove the bound for the windowed and the blocked patterns. In \ref{supp:bound-strided}, we argue the optimality for the strided pattern.

\begin{theorem}
    Suppose we have a strided pattern with stride $X$. $\lambda^s$ is the number of thread-blocks generated by poset-tiling with stretch factor $s$. Then we have that $\lambda^s$ only decreases when $\gcd(s, X)$ increases.
\end{theorem}

We prove this by computing $\lambda^s$ as a monotonically decreasing function of $gcd(s, X)$.

\subsection{Structured Polygonal Patterns}
\label{supp:bound-poly}

\begin{definition}
    A mask $P$ is considered a structured dense polygon if a significantly large horizontal subsection of the mask consisting of $rh$ rows ($N - rh \ll N$) can be partitioned as $r$ repeating slabs of shape $h \times l$ recursively shifted by $l'$ columns.  \\
    $$ P = \bigg\{ \bigg( \bigg\lfloor \frac{y}{h} \bigg\rfloor l' + t,~ y \bigg) ~\bigg|~ t \in \mathbb{Z}_l,~ y \in \mathbb{Z}_{rh} \bigg\} $$
\end{definition}

While the mathematical definition looks restrictive, several transformers have sampling masks that can be parameterized by this definition. For example, the block pattern had $h=l'=w$ and $l=2w$ while the window pattern had $h=l'=1$ and $l=2w+1$.

\begin{definition}
    Let us consider a Naive Tiling algorithm that partitions the mask $P$ into patches where a patch is a region of $m$ consecutive rows. It then tiles these patches from left to right. \\
    We define $w^{(k)}$ as the width of the $k^{th}$ patch. If a patch overlaps with $t ~(\ge 1)$ slabs, its width is $l + (t - 1) l'$.
\end{definition}

\begin{theorem}
    Given a structured dense polygon mask $P$ with parameters $(r, l, h, l')$ such that $gcd(m, h) = \kappa$ (i.e., $m = \tau_m \kappa$, $h = \tau_h \kappa$) and $\tau_m = \alpha \tau_h + \beta$, $\beta \in \mathbb{Z}_{\tau_h}$, the Naive Tiling Algorithm will generate a TB of size $\lambda_{naive}$.
    \begin{equation*}
        \lambda_{naive} =
        \begin{cases}
        \bigg( \frac{r}{\tau_m} \bigg) \bigg( \tau_h \bigg\lceil \frac{l + (\alpha - 1) l'}{n} \bigg\rceil \bigg) & \beta = 0 \text{ (i.e., } h | m) \\
        \bigg( \frac{r}{\tau_m} \bigg) \bigg( (\beta - 1) \bigg\lceil \frac{l + (\alpha+1) l'}{n} \bigg\rceil ~+~ (\tau_h - \beta + 1) \bigg\lceil \frac{l + \alpha l'}{n} \bigg\rceil \bigg) & \beta \neq 0
        \end{cases}
    \end{equation*}
\end{theorem}

\begin{proof}
    $P$ ($rh$ rows) can be partitioned into $\frac{r}{\tau_m}$ horizontal cross-sections of $P$ with $lcm(m, h) = \tau_m \tau_h \kappa$ rows. \\
    Each horizontal cross-section will be covered by $\tau_h$ patches ($m$ rows) and $\tau_m$ slabs ($h$ rows). \\
    Therefore, $\lambda_{naive} = (\frac{r}{\tau_m}) (\sum_{i=1}^{\tau_h} \lceil \frac{w^{(i)}}{n} \rceil)$.

    If $\beta = 0$, i.e. $h | m$, we have $\kappa = h$, $\tau_m = \alpha$, $\tau_h = 1$. The cross-section is covered by 1 patch overlapping with $\tau_m$ slabs with width $(l + (\tau_m - 1) h)$ $[= (l + (\alpha - 1) l')]$. This proves the first case of the equation.

    If $\beta \neq 0$ and $\alpha = 0$ (i.e., $m < h$), we have $\tau_m = \beta$. Every patch will overlap with either 1 slab or 2 slabs. Since there are $\tau_m$ slabs, exactly $(\tau_m - 1)$ $[= (\beta - 1)]$ patches will overlap with 2 slabs and will have width $(l + l')$ $[= (l + (\alpha + 1) l')]$. Rest of the $(\tau_h - \tau_m + 1)$ $[= (\tau_h - \beta + 1)]$ patches will overlap with 1 slab and will have width $l$ $[= (l + \alpha l')]$. This proves the second case of the equation for $\alpha = 0$.

    If $\beta \neq 0$ and $\alpha \neq 0$ (i.e., $m > h$), every patch will overlap with either $(\alpha+1)$ slabs or $(\alpha+2)$ slabs. Since there are $\tau_h$ patches, exactly $(\tau_h - 1)$ slabs will overlap with 2 patches. \\
    Say $x$ patches overlap with $(\alpha + 2)$ slabs. We get $(x)(\alpha + 2) + (\tau_h - x)(\alpha + 1) = \tau_m + (\tau_h - 1)$. Solving for $x$ gives us $x = (\tau_m - \alpha \tau_h) - 1 = \beta - 1$. \\
    There are $(\beta - 1)$ patches that will overlap with $(\alpha + 2)$ slabs and will have width $(l + (\alpha + 1) l')$. Rest of the $(\tau_h - \beta + 1)$ patches will overlap with $(\alpha + 1)$ slab and will have width $(l + \alpha l')$. This proves the second case of the equation for $\alpha \neq 0$.
\end{proof}

\begin{theorem}
    For a given structured dense polygon mask $P$ with parameters $(r, l, h, l')$, Algorithm A and Tiling Algorithm generate $TB_A$ and $TB_{naive}$. We define $\lambda^{(k)}$ and $\lambda_{naive}^{(k)}$ as the number of thread blocks with anchored in first $k$ patches (km rows) for $TB_A$ and $TB_{naive}$. $\delta^{(k)}$ thread blocks in $TB_A$ are anchored in the first $k$ patches and aren't contained within the first $k$ patches, in other words, thread blocks which overflow. \\
    For any $k \ge 1$, $\lambda^{(k)} - \delta^{(k)} \ge \lambda_{naive}^{(k)} - \sum_{i=1}^{k} \lceil \frac{w^{(i)} - l}{n} \rceil$.
\end{theorem}

\begin{proof}
    We will prove the theorem by induction. \\
    We know that $\lambda_{naive}^{(k)} = \sum_{i=1}^{k} \lceil \frac{w^{(k)}}{n} \rceil$.
    
    Base Case: We need to show that $\lambda^{(1)} - \delta^{(1)} \ge \lambda_{naive}^{(1)} - \lceil \frac{w^{(1)} - l}{n} \rceil$. \\
    The minimum value of $\lambda^{(1)} - \delta^{(1)}$ is $\lceil \frac{l}{n} \rceil$ since the first row of the first patch needs to be covered by $TB_A$. \\
    \begin{equation*}
        \begin{aligned}
            \lambda^{(1)} - \delta^{(1)} &\ge \bigg\lceil \frac{l}{n} \bigg\rceil \\
                                         &\ge \bigg\lceil \frac{w^{(1)}}{n} \bigg\rceil - \bigg\lceil \frac{w^{(1)} - l}{n} \bigg\rceil &(\text{Since, } \lceil x \rceil + \lceil y \rceil \ge \lceil x+y \rceil ) \\
                                         &= \lambda_{naive}^{(1)} - \bigg\lceil \frac{w^{(1)} - l}{n} \bigg\rceil
        \end{aligned}
    \end{equation*}

    Hypothesis: Assuming $k \ge 1$, such that $\lambda^{(k)} - \delta^{(k)} \ge \lambda_{naive}^{(k)} - \sum_{i=1}^{k} \lceil \frac{w^{(i)} - l}{n} \rceil$.

    Inductive Step: We need to show that $\lambda^{(k+1)} - \delta^{(k+1)} \ge \lambda_{naive}^{(k+1)} - \sum_{i=1}^{k+1} \lceil \frac{w^{(i)} - l}{n} \rceil$. \\
    The minimum value of $\lambda^{(k+1)} - \delta^{(k+1)}$ is $\lambda^{(k)} + \lceil \frac{l - \delta^{(k)} n}{n} \rceil$ since the first row of the $(k+1)^{th}$ patch needs to be covered by $TB_A$. The only thread blocks that can cover this row are ones with the anchor in the first row of $(k+1)^{th}$ patch and the ones with the anchor in the $k^{th}$ patch but overflow.
    \begin{equation*}
        \begin{aligned}
            \lambda^{(k+1)} - \delta^{(k+1)} &\ge \lambda^{(k)} + \bigg\lceil \frac{l - \delta^{(k)} n}{n} \bigg\rceil \\
                                             &= \lambda^{(k)} - \delta^{(k)} + \bigg\lceil \frac{l}{n} \bigg\rceil \\
                                             &\ge \lambda_{naive}^{k} - \sum_{i=1}^{k} \bigg\lceil \frac{w^{(i)} - l}{n} \bigg\rceil + \bigg\lceil \frac{l}{n} \bigg\rceil &(\text{Using the Hypothesis}) \\
                                             &\ge \lambda_{naive}^{k} - \sum_{i=1}^{k} \lceil \frac{w^{(i)} - l}{n} \rceil + \bigg\lceil \frac{w^{(k+1)}}{n} \bigg\rceil - \bigg\lceil \frac{w^{(k+1)} - l}{n} \bigg\rceil &(\text{Since, } \lceil x \rceil + \lceil y \rceil \ge \lceil x+y \rceil ) \\
                                             &= \lambda_{naive}^{(k)} + \bigg\lceil \frac{w^{(k+1)}}{n} \bigg\rceil - \sum_{i=1}^{k+1} \bigg\lceil \frac{w^{(i)} - l}{n} \bigg\rceil \\
                                             &= \lambda_{naive}^{(k+1)} - \sum_{i=1}^{k+1} \bigg\lceil \frac{w^{(i)} - l}{n} \bigg\rceil
        \end{aligned}
    \end{equation*}
\end{proof}

\begin{theorem}
    For a given structured dense polygon mask $P$ with parameters $(r, l, h, l')$ such that $l' \approx h, l \gg m$, say the minimum number of thread blocks required to cover is $\lambda_{opt}$.
    \begin{equation*}
        \frac{\Delta \lambda}{\lambda_{opt}} = \frac{\lambda_{naive} - \lambda_{opt}}{\lambda_{opt}} \le \frac{m}{l}
    \end{equation*}
\end{theorem}

\begin{proof}
    We know that $\lambda_{opt} \ge \frac{Nl}{mn}$. \\
    From the previous theorem, we get that $\sum_{i=1}^{(N/m)} \lceil \frac{w^{(i)} - l}{n} \rceil \ge \lambda_{naive} - \lambda_{opt}$.

    When $h < m$, we have $\sum_{i=1}^{(N/m)} \lceil \frac{w^{(i)} - l}{n} \rceil \le \frac{N}{m} \frac{(\alpha + 1) l'}{n}$. \\
    Therefore, $\frac{\lambda_{naive} - \lambda_{opt}}{\lambda_{opt}} \le \frac{(\alpha + 1) l'}{l} \approx \frac{m}{l}$

    When $h > m$, we have $(\tau_m - 1)$ patches that overlap 2 slabs for every $\tau_h$ patches. In other words, for every $\tau_h$ patches, there are $(\tau_m - 1)$ patches such that $w^{(i)} - l$ is $l'$ while the rest of patches have $w^{(i)} - l$ as 0. \\
    Therefore, $\sum_{i=1}^{(N/m)} \lceil \frac{w^{(i)} - l}{n} \rceil = \frac{N}{m} \frac{(\tau_m-1) l'}{\tau_h n} \le \frac{N}{m} \frac{\tau_m l'}{\tau_h n} = \frac{Nl'}{hn}$. \\
    Therefore, $\frac{\lambda_{naive} - \lambda_{opt}}{\lambda_{opt}} \le \frac{l'm}{hl} \le \frac{m}{l}$
\end{proof}

Finally, we get $\frac{\lambda_{naive}}{\lambda_{opt}} \leq 1 + \frac{m}{l}$. We note that the Naive Tiling Algorithm is essentially the Poset Tiling Algorithm applied to each patch. Thus, we can argue that $\lambda_{poset} \le \lambda_{naive}$.

This gives us $\frac{\lambda_{poset}}{\lambda_{opt}} \leq 1 + \frac{m}{l}$. As explained in \S\ref{supp:stretch-poly}, the optimal stretch factor for polygonal patterns is 1, so we get $\phi_{CMR} = 1$.

Thus, we prove that $\frac{Cost(TB_{poset})}{Cost(TB_{opt})} \leq 1 + \frac{m}{l}$.

\subsection{Strided Patterns}
\label{supp:bound-strided}

By definition of strided pattern and poset tiling algorithm, we have no redundant compute ($\phi_{R}$) and minimum thread-divergence ($\phi_{TD}$). Thus, for a given stretch $s$, $\lambda_{opt} = \lambda_{poset}$. The stretch factor selection algorithm searches the space of the stretch factors and selects the stretch factor that minimizes our cost model. Thus, trivially, $Cost(TB_{opt}) = Cost(TB_{poset})$.

\end{document}